\documentclass[lettersize,journal]{IEEEtran}
\usepackage{amsmath,amsfonts}
\usepackage{array}
\usepackage[caption=false,font=normalsize,labelfont=sf,textfont=sf]{subfig}
\usepackage{textcomp}
\usepackage{stfloats}
\usepackage{url}
\usepackage{verbatim}
\usepackage{graphicx}
\def\BibTeX{{\rm B\kern-.05em{\sc i\kern-.025em b}\kern-.08em
    T\kern-.1667em\lower.7ex\hbox{E}\kern-.125emX}}
\usepackage{balance}

\usepackage{cite}
\usepackage{amsmath,amssymb,amsfonts}
\usepackage{amsthm}
\usepackage{graphicx}
\usepackage{textcomp}
\usepackage{xcolor}
\usepackage{hyperref}
\usepackage{booktabs}
\usepackage[nameinlink,capitalize]{cleveref}
\usepackage{algorithm}
\usepackage{algpseudocode}

\theoremstyle{plain}
\newtheorem{assumption}{Assumption}
\theoremstyle{plain}
\newtheorem{lemma}{Lemma}
\theoremstyle{plain}

\theoremstyle{plain}
\newtheorem{remark}{Remark}
\theoremstyle{plain}
\newtheorem{theorem}{Theorem}
\theoremstyle{plain}

\theoremstyle{plain}
\newtheorem{definition}{Definition}

\crefname{assumption}{Assumption}{Assumptions}
\crefname{lemma}{Lemma}{Lemma}

\allowdisplaybreaks[4]

\begin{document}


{\title{\color{black} Three Birds, One Stone: Solving the Communication-Memory-Privacy Trilemma in LLM Fine-tuning Over Wireless Networks with Zeroth-Order Optimization
}

\author{
        Zhijie Cai, Yuhao Zheng, Haolong Chen, Dongzhu Liu, Bin Wang, Guangxu Zhu
\thanks{Z. Cai and H. Chen are with Shenzhen International Center for Industrial and Applied Mathematics, Shenzhen Research Institute of Big Data, The Chinese University of Hong Kong-Shenzhen, Guangdong, China.} 
\thanks{Y. Zheng is with the School of Information and Communication Engineering, Beijing University of Posts and Telecommunications, Beijing 100876, China. }
\thanks{D. Liu is with the School of Computing Science, University of Glasgow, Glasgow, U.K.}
\thanks{B. Wang and G. Zhu are with Shenzhen Research Institute of Big Data, The Chinese University of Hong Kong-Shenzhen, Guangdong, China.}
\thanks{\textit{Corresponding authors: B. Wang (wangbin@sribd.cn) and G. Zhu (gxzhu@sribd.cn).}}
}



\maketitle

\begin{abstract}
\textcolor{black}{Federated Learning (FL) offers a promising pathway for collaboratively fine-tuning Large Language Models (LLMs) at the edge; however, this paradigm faces a critical bottleneck: the prohibitive communication and memory overheads incurred by exchanging high-dimensional gradients. Furthermore, recent studies reveal that user training data can still be recovered from these local gradients, undermining the core privacy promise of FL. In this paper, we address this trilemma of communication, memory, and privacy by proposing \emph{pAirZero}, a novel framework that synergizes Zeroth-Order (ZO) optimization with Over-the-Air (OTA) computation. {\color{black} Uniquely, \emph{pAirZero} enables resource-constrained devices to submit their local gradient with only bit-level communication loads while participating in federated fine-tuning of LLMs with inference-level memory costs. This approach not only eliminates the high memory requirements needed for LLM fine-tuning but also alleviates the strict synchronization requirements that plague conventional OTA methods. } We further formulate a rigorous optimization model to {\color{black} adaptively} determine the optimal transmit power and noise levels, ensuring consistent privacy protection regardless of channel conditions. Numerical experiments demonstrate the superiority of \emph{pAirZero} in enabling secure, efficient LLM fine-tuning over wireless networks, {\color{black} with only $25\%$ peak memory cost on OPT-125M and magnitudes-of-order lower communication load compared to conventional methods}.}
\end{abstract}


\section{Introduction}
The advent of Large Language Models (LLMs) has revolutionized artificial intelligence, enabling unprecedented capabilities in natural language understanding and generation. However, the sheer scale of these models—often spanning billions of parameters—traditionally confines their deployment to centralized data centers. This centralization creates a critical disconnect: while the most valuable, context-rich data resides at the network edge (on smartphones, IoT devices, and vehicles), the computational power required to learn from it remains locked in the cloud. {\color{black} In contrast, recent regulations such as the GDPR (General Data Protection Regulation) and the ADPPA (American Data Privacy and Protection Act) have imposed restrictions on sharing privacy-sensitive data among different clients or platforms. Consequently, collecting a sufficient amount of data to train machine learning models for these applications is a challenging task. The gap becomes even more emergent, as pointed out by \cite{villalobos2024will}, publicly available data will soon be depleted. Therefore, it is crucial to enable the secure use of personal data for model training. } Federated Learning (FL) has emerged as the de facto solution to bridge this gap, allowing edge devices to collaboratively fine-tune models without exposing raw private data.

Yet, deploying FL for LLMs at the edge introduces a formidable ``trilemma" that current infrastructure struggles to support:

\begin{itemize}
\item \textbf{Communication Bottleneck:} Transmitting the gradients of massive LLMs overwhelms the limited bandwidth of wireless edge networks.
\item \textbf{Memory Wall:} Standard backpropagation (BP) requires storing intermediate activation maps, demanding memory far exceeding the capacity of typical edge devices.
\item \textbf{Privacy Leakage:} Despite FL's promise, recent gradient inversion attacks \cite{zhu2019deep} have demonstrated that sensitive user data can be reconstructed from the very gradients intended to protect it.
\end{itemize}

Existing solutions often tackle these challenges in isolation. Quantization \cite{bernstein2018signsgd} and sparsification \cite{stich2018sparsified} reduce communication but overlook memory constraints. Split learning alleviates memory pressure but incurs high latency \cite{thapa2022splitfed}, and Differential Privacy (DP) mechanisms often degrade model accuracy \cite{dwork2006differential, abadi2016deep}. Furthermore, traditional Over-the-Air (OTA) computation—a physical-layer technique that aggregates signals via channel superposition—offers a potential solution to bandwidth constraints \cite{zhu2019broadband}; however, it suffers from strict synchronization requirements \cite{pradhan2025experimental} and analog signal distortion \cite{yang2021revisiting}.

In this paper, we propose a unified framework that leverages these challenges to create synergies. We observe that Zeroth-Order (ZO) optimization, which estimates gradients by probing the loss function with random perturbations, naturally aligns with the properties of OTA computation. ZO is inherently robust to the noise that plagues analog OTA transmission, turning the wireless channel's imperfections into a mechanism for differential privacy. Furthermore, by eliminating the need for backpropagation, ZO optimization shatters the memory wall; by leveraging the superposition property of wireless channels, OTA computation dissolves the communication bottleneck. 
{\color{black} Ultimately, a classic ZO algorithm, known as SPSA \cite{spall1992multivariate}, if combined with PRNG \cite{salmon2011parallel}, can reduce the per iteration communication load to bit-level, as will be detailed in \cref{sec:learning}.}

Building on this insight, we introduce \emph{pAirZero}, a novel privacy-preserving FL framework designed specifically for fine-tuning LLMs at the edge. Unlike traditional approaches that treat privacy as an afterthought, \emph{pAirZero} embeds DP directly into the wireless transmission process. We utilize the inherent noise of wireless channels, supplemented by artificial noise injection, to mask individual contributions "in the air." This creates a privacy-by-design architecture where the aggregation process itself acts as the privacy mechanism.

Our specific contributions are as follows:

\begin{itemize}
\item \textbf{Holistic Framework for Efficient Edge LLM Fine-Tuning:} We propose \emph{pAirZero}, the first framework to synergize ZO optimization with OTA computation for {\color{black} private} LLM fine-tuning. This combination {\color{black} natively supports a wider range of indifferentiable objectives,} reduces memory consumption to inference-level standards {\color{black}($75 \%$ reduction)} and decouples communication costs from the number of participating devices. Moreover, ZO reduces the per-iteration load to a bit-level by refraining from sending true gradients, instead sending gradient projections on pseudo-random directions.
\item \textbf{Privacy-by-Design Transmission:} We develop a rigorous differential privacy mechanism embedded within the gradient transmission. By optimizing the transmit power and injecting calibrated artificial noise, we ensure that the aggregated signal remains useful for learning while mathematically guaranteeing that individual user data cannot be distinguished.
\item \textbf{Empirical Validation:} Extensive experiments on the {\color{black} well-recognized and light-weighted} OPT-125M model demonstrate that \emph{pAirZero} achieves test {\color{black} performances} comparable to that of ideal, non-private baseline while significantly outperforming standard methods {\color{black}and naive schemes }in terms of communication and memory efficiency.
\end{itemize}

The remainder of this paper is organized as follows. We first review related works in \cref{sec:related}, introduce the system model in \cref{sec:mod&def}, and present the algorithm design in \cref{sec:alg}. 
We conduct convergence bound analysis in \cref{convergence-pairzero}, and minimize it with optimization problems in \cref{optmini}.
We present empirical evaluations \cref{sec:sim} and conclude this work in \cref{sec:conclude}.

\section{Related Works}\label{sec:related}

\subsection{Communication-Efficient FL} 
The benefits of FL come at the expense of extremely high communication costs, as FL requires frequent gradient uploads \cite{konevcny2016federated}, 
The communication {\color{black} bottleneck} of FL originates in the excessively high dimensionality of the gradient vectors.
This bottleneck becomes even more prominent when training or fine-tuning LLMs. Compression-based methods offer a way to reduce the per-iteration communication cost in FL. The first direction is to sparsify \cite{stich2018sparsified} the local gradient such that only a small portion of elements need to be uploaded to the edge server. The rationale behind this is that most of the gradient elements are small (in magnitude), and reserving only those significant elements can well approximate the direction of the local gradient. A similar strategy is to let each client upload a quantized local gradient \cite{bernstein2018signsgd}
. In addition to element-wise compression, client selection-based methods \cite{chen2018lag} and local training-based methods \cite{pmlr-v54-mcmahan17a, zhang2021fedpd, mishchenko2022proxskip} can also be regarded as special types of compression-based methods. The rationale behind these methods is to reduce the upload frequency of each client, thereby alleviating the communication burden in each iteration; however, this often occurs at the cost of slower convergence.

{\color{black} Apart from machine learning algorithm design, advanced communication protocols can also be employed to improve communication efficiency.}
Orthogonal multiple access (OMA) allows numerous clients to share the same communication resource (such as time, frequency, or code) by assigning mutually orthogonal signal resources to different clients, thereby enabling the receiver to avoid interference between clients. However, the number of clients that can simultaneously be served is limited by the total amount of orthogonal resources (such as time slots and subcarriers). In massive connectivity scenarios such as federated learning, this becomes a bottleneck. An alternative to OMA is the so-called non-orthogonal multiple access (NOMA). In NOMA, clients share the same time-frequency resources, with differentiation achieved through power differences (power-domain NOMA) \cite{maraqa2020survey} or sparse codebooks (code-domain NOMA) \cite{liu2021sparse}. Theoretically, NOMA can support overloaded transmission scenarios in which the number of clients exceeds the number of available resource blocks. FL applications favor such a characteristic \cite{sun2020adaptive}. However, NOMA is affected by the phenomenon of error propagation. Specifically, if one client's signal is incorrectly decoded, the error will accumulate, causing all subsequent decoding attempts to fail as well. In conventional FL, the edge server averages all local gradients after receiving each individual. 

Similar to NOMA, OTA also allows all clients to share the same radio resources. The key difference is that the gradient average is obtained through the electromagnetic superposition of all clients' wireless signals. Therefore, OTA does not require decoding each client's signal. In OTA, all client transmits their local gradients synchronously with uncoded transmission \cite{zhu2019broadband,sery2020analog,amiri2020machine,zhang2021gradient,yang2020federated}. The performance can be further enhanced with proper power control \cite{cao2020optimized}, beamforming design \cite{yao2025over}, and inference management \cite{yao2025energy}. While these works assume analog modulations, \cite{zhu2020one} established a digital modulation-based framework. However, all the above-mentioned methods require stringent transmit synchronization across different clients, which is hardly achievable in practical systems (see \cite{pradhan2025experimental} for an experimental implementation). 

\subsection{{\color{black} Memory-Efficient FL}}
{\color{black} We notice that there are some overlapping techniques between memory-efficient and communication-efficient FL. In the particular area of fine-tuning, Parameter-Efficient Fine-Tuning (PEFT) \cite{Adapter, Prefix-tuning, LoRA} is deeply related to compression-based methods. For example, the essence of Low Rank Adaptation (LoRA) \cite{LoRA} is to freeze the weights of the pre-trained model and then update several additional low-rank trainable parameter matrices. Therefore, this method can be easily incorporated into FL to reduce communication and memory costs, as seen in \cite{sun2024improving, cho2023heterogeneous, zhang2023fedpetuning}. 
Nevertheless, these methods still face communication bottlenecks because the communication load still scales with the number of learnable parameters. Moreover, they still rely on first-order (FO) optimization methods that require excessive memory, such as SGD and Adam. 
}

Recently, the approach of zeroth-order (ZO) optimization \cite{spall1992multivariate}
is regarded as a possible solution to overcome this {\color{black} bottleneck}. Instead of directly computing the gradient via backpropagation, ZO estimates the gradient by calculating one or a few random directional derivatives. The effectiveness of ZO-SGD for fine-tuning LLMs in a centralized setting is first demonstrated in \cite{malladi2023fine}. It shows that ZO-SGD performs comparably to full-parameter fine-tuning across various tasks while reducing memory consumption by up to 12 times \cite{malladi2023fine}. In this regard, uploading the estimated gradient can be simplified to transmitting only the {\color{black} gradient projection on the designated random direction coded by the random seed}. Subsequently, ZO-based methods have been rapidly adopted in federated learning \cite{ZO-strength-1}
, with the focus shifting from memory reduction to improving communication efficiency. More recently, studies such as FedKSeed \cite{FedKSeed} and FwdLLM \cite{FwdLLM} have demonstrated that, empowered by ZO, even low-capacity networks with transmission rates of only a few kilobytes per second can support federated LLM fine-tuning. Additionally, a recent pioneering work \cite{chen2025zeroth} combines OTA and ZO to enhance the spectrum efficiency further.

\begin{figure*}[t]
    \centering
    \includegraphics[width=0.9\linewidth]{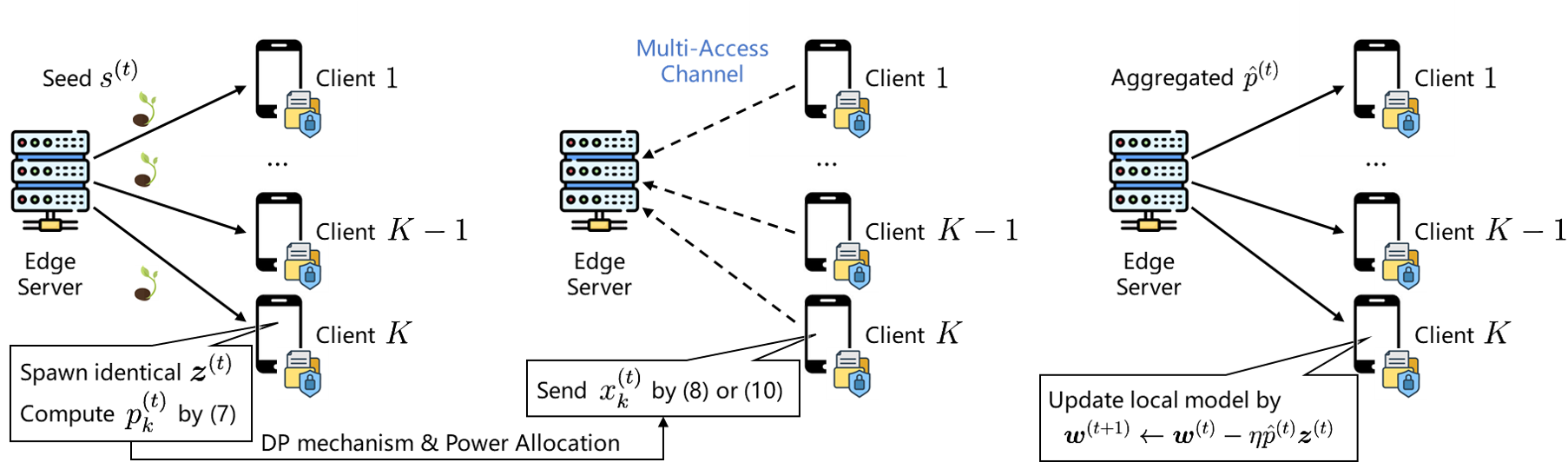}
    \caption{The workflow of \emph{pAirZero}.}
    \label{fig:placeholder}
\end{figure*}

\subsection{Differentially Private FL} 
Even though FL avoids direct data exposure, it is still possible to recover client data from the uploaded local gradient (see \cite{zhu2019deep}). Malicious clients could potentially infer the presence of an individual data sample from a learned model through a membership attack \cite{shokri2017membership} or a model inversion attack \cite{fredrikson2015model}. Differential privacy provides a cheap and convenient way to enhance privacy protection. The core idea of DP is to introduce uncertainty into the disclosed information, thereby obscuring contributions from individual data points. A considerable amount of research has been conducted on the application of DP in FL, including standard FL \cite{wei2020federated}, NOMA-based FL \cite{seif2020wireless}, digital quantization-based FL \cite{sonee2020efficient}, and OTA-based FL \cite{liu2020privacy}. However, these works all employ traditional back-propagation for model training. More recently, DP has also been combined with ZO-based FL methods, as seen in \cite{zhang2023dpzero, liu2024differentially, tang2024private}. Nevertheless, these studies do not account for imperfect communications. Therefore, differential privacy cannot guarantee a consistent level of privacy protection.

\section{{\color{black} System Model}}\label{sec:mod&def}

\subsection{Learning Model}
We consider a wireless federated edge learning system consisting of a single-antenna edge server and $K$ single-antenna clients. Each client is connected to the edge server via a shared noisy channel. We aim to fine-tune an LLM collaboratively; this amounts to solving the following problem: 
\begin{align}\label{eq:aim}
    \min\limits_{\boldsymbol{w} \in \mathbb{R}^d} F(\boldsymbol{w}), 
\end{align}
where $\boldsymbol{w}\in\mathbb{R}^d$ is the model vector to be fine-tuned, 
\begin{align}
    \textstyle\quad F(\boldsymbol{w}) \triangleq \frac{1}{K} \sum_{k=1}^K F_k(\boldsymbol{w}),
\end{align}
is the global loss function, 
\begin{align}
     \textstyle F_k(\boldsymbol{w}) \triangleq \frac{1}{|\mathcal{D}_k|} \sum_{\{\boldsymbol{u}_{k, i}^{(t)}, v_{k, i}^{(t)}\}\in\mathcal{D}_k} f(\boldsymbol{w}; \boldsymbol{u}_{k, i}^{(t)}, v_{k, i}^{(t)}),
\end{align}
is the local loss function held by client $k$, $\mathcal{D}_k$ is the client $k$'s local training dataset, and $\{\boldsymbol{u}_{k, i}^{(t)}, v_{k, i}^{(t)}\}$ is the $i$th training data sample. 

We assume a SGD-based wireless federated edge learning process, specifically, each client $k$ estimate the gradient of its local loss function $\tilde{\boldsymbol{g}}^{(t)}_k$ at each iteration $t$, then the clients update their models by $\boldsymbol{w}^{(t+1)} = \boldsymbol{w}^{(t)} - \eta \frac{1}{K}\sum_{k=1}^K \tilde{\boldsymbol{g}}^{(t)}_k $.

\subsection{Communication Model}
The essence of OTA transmission is to obtain an average of several transmitted numerical values via the superposition of electromagnetic waves emitted by the clients, each encoded by $p_k^{(t)}$. Consider a block fading channel where the channel coefficient remains unchanged within the $t$-th transmission, the received signal at the edge server can be written as 
\begin{align}
   \textstyle y^{(t)} = \sum_{k=1}^K h_k^{(t)} x_k^{(t)}+z^{(t)},
\end{align}
where $h_k^{(t)}\in\mathbb{C}$ is the channel coefficient between the $k$th client and the edge server, $x_{k}^{(t)}\in\mathbb{C}${\color{black}, a processed version of $p_k^{(t)}$}, is client $k$'s uploading signal (containing both the desired signal and artificial noise), $z^{t}\in\mathbb{C}$ is the Gaussian random noise at the edge server. {\color{black} We defer the specific form of $x_k^{(t)}$ to the following sections.}

Upon receiving $y^{(t)}$, rather than recovering the noiseless average value, the edge server attempts to recover the average $\frac{1}{K}\sum_{k=1}^K p_k^{(t)}$ through channel inversion, namely,  
\begin{align}
    \textstyle \hat{p}^{(t)} = \frac{y^{(t)}}{K c^{(t)}}. \label{eq:estp}
\end{align}
After this, the noisy average is broadcast back to all clients\footnote{We assume that the downlink broadcast is noise-free. This is reasonable since the edge server is much more communication powerful than the clients.}. Upon receiving $\hat{p}^{(t)}$, the clients will get ready for a next round of computation.


\subsection{Trillema of Communication, Memory, and Privacy}

In conventional FO-based edge learning systems, {\color{black} we face significant communication costs when setting $p_k^{(t)}$ as the elements of the gradient vectors $\tilde{\boldsymbol{g}}_k^{(t)}$ computed by BP, which also incurs substantial memory consumption.} Moreover, privacy-preserving by DP introduces additional computation by generating and adding correct noise vectors to the original gradient, resulting in a significant computational load on resource-constrained devices. This incentivizes us to find a communication-memory-privacy-efficient way for edge LLM fine-tuning.

\section{{\color{black} \emph{pAirZero}: Communication-Memory-Privacy-\\Efficient Edge LLM Fine-tuning}}\label{sec:alg}

{\color{black} As done in most of prior works, the payload $x_k^{(t)}$ is the true estimated gradients from back-propagations. However, in order to improve the availability of LLMs on resource-constrained devices, we switch to ZO-based optimization. We first describe the local gradient estimation, then we elaborate on the details of transmission. We then confirm the feasibility of the design with convergence analyses. An illustration of the workflow is \cref{fig:placeholder}.}

\subsection{\color{black} Memory-efficient Gradient Estimation}
\label{sec:learning}
At each iteration, each client $k$ pulls a seed $s^{(t)}$ from the edge server and then estimates its local gradient by simultaneous perturbative stochastic approximation (SPSA): 
\begin{align}\label{eq:localg}
 \textbf{\text{Local gradient:}}\quad \boldsymbol{g}_k^{(t)} = p_k^{(t)}\boldsymbol{z}^{(t)},
\end{align}
where $\boldsymbol{z} \in \mathcal{N}(\boldsymbol{0}, \boldsymbol{I}_d)$ is random vector generated by the seed $p_k^{(t)}$, $\mu$ is the scale of perturbation, and $p_k^{(t)}$ is referred to as the gradient projection which is defined as 
\begin{align}\label{eq:projgrad}
    \textbf{\text{Local gradient projection:}} \ \textstyle p_k^{(t)} = \frac{F_k(\boldsymbol{w} + \mu \boldsymbol{z}^{(t)}) - F_k(\boldsymbol{w} - \mu \boldsymbol{z}^{(t)})}{2 \mu}.
\end{align}
In the above, $F_k(\boldsymbol{w} + \mu \boldsymbol{z}^{(t)})$ and 
$F_k(\boldsymbol{w} - \mu \boldsymbol{z}^{(t)})$ are computed using either the entire local training data or a mini-batch of local training data. In ZO, only $p_k^{t}$ needs to be uploaded to the edge server for gradient aggregation. This is because $\boldsymbol{z}^{(t)}$ can be generated by the random seed $s^{(t)}$ stored at the edge server. With each local gradient given by \cref{eq:localg}, the global aggregated gradient should be $\boldsymbol{g}^{(t)}=(\frac{1}{K}\sum_{k=1}^K p_k^{(t)})\boldsymbol{z}^{(t)}$. 

{\color{black}
\begin{remark}[Compatibility with indifferentiable objectives]
Unlike in FO methods, it is not necessary for $F_k$ to be a differentiable function. This will give ZO methods an advantage in dealing with a wider range of objective functions that are more often used in modern LLMs.
\end{remark}
}

{\color{black}
\begin{remark}[Memory efficiency]
In the above, instead of adopting back-propagation to compute a true gradient as done in most prior works, we employ a ZO gradient estimation method. Unlike acquiring excessive memory for backpropagation, the ZO method only requires inference-level memory overhead, as the gradient computation relies solely on loss function values. This will significantly reduce the memory overhead during the fine-tuning process, making it exceptionally appealing for clients (often assumed to be resource-constrained) participating in the FL system. A simple comparison can be found in \cref{memory-communication-efficiency}.
\end{remark}
}

\subsection{{\color{black} Communication-efficient Gradient Aggregation}}


{\color{black} In terms of communication, we consider two different designs: \emph{pAirZero} and \emph{Sign-pAirZero} as follows.}

\subsubsection{{\color{black} \emph{pAirZero}}}
With analog modulation, the sending signal is 
\begin{align}
x_k^{(t)}=\alpha_k^{(t)} (p_k^{(t)} + n_k^{(t)}), \label{eq:tx_analog}
\end{align}
where $\alpha_k^{(t)}$ is the scaling factor, $n_k^{(t)}\in\mathbb{C}$ is the artificial noise used to facilitate differential privacy. In this regard, the received signal at the edge server becomes 
\begin{align}\label{eq:rec_analog}
   \textstyle y^{(t)} = \sum_{k=1}^K h_k^{(t)} (\alpha_k^{(t)} (p_k^{(t)} + n_k^{(t)}))+z^{(t)},
\end{align}
where $\alpha_k^{(t)}$ is the scaling factor. 
\subsubsection{{\color{black} \emph{Sign-pAirZero}}}
{\color{black} To accommodate digital modulation, which is more similar to modern communication design, we {\color{black} utilize a one-bit compression on the gradient projection for simplicity.}} We then design the sending signal to} \begin{align}
x_k^{(t)}=\alpha_k^{(t)} (\text{sign}\big\{\sum_{k=1}^K \text{sign}\{p_k^{(t)}\}\big\} + n_k^{(t)}), \label{eq:tx_digital}
\end{align} 
and thus the received signal is given as 
\begin{align}\label{eq:rec_digital}
     \textstyle y^{(t)} = \sum_{k=1}^K h_k^{(t)} (\alpha_k^{(t)} (\text{sign}\{p_k^{(t}\} + n_k^{(t)}))+z^{(t)},
\end{align}

Suppose $h_k^{(t)}\alpha_k^{(t)}=c^{(t)}$, then for both cases, the effective noise is $\sum_{k=1}^K c^{(t)}n_k^{(t)}+z^t$, whose standard deviation is given as 
\begin{align}
   \textstyle m^{(t)} =\big((c^{(t)})^2 \sum_{k=1}^K (\sigma_k^{(t)})^2 + N_0\big)^{1/2}.\label{eq:m}
\end{align}

\begin{remark}[{\color{black} Communication efficiency}]
    The communication efficacy of \emph{pAirZero} is two-fold. Firstly, the above 1-bit compression seems to share a similar form with signSGD \cite{bernstein2018signsgd}. However, it should be noted that signSGD compresses each local gradient element-wise. For this reason, the communication overhead scales with the ambient dimension. In contrast, both \emph{pAirZero} and \emph{Sign-pAirZero} admit an $\mathcal{O}(1)$ communication load independent of the ambient model dimension. {\color{black} This results in magnitudes of order lower communication overhead, and will alleviate the stringent communication synchronization often assumed in previous OTA computation works, making OTA computation more realistic for resource-constrained clients.} Secondly, note that OTA computation allows all users to transmit on the same resource block, the latency is independent of the number of participating clients. 
\end{remark}

\subsubsection{{\color{black} Model Updates}}\label{subsubsec:model updates}
Each client updates the model vector via $\boldsymbol{w}^{(t+1)}=\boldsymbol{w}^{(t)}-\eta \cdot \hat{p}^{(t)}\boldsymbol{z}^{(t)}$, where $\eta$ is the learning rate.

\subsection{{\color{black} Privacy-preserving Gradient Sharing}}
Differential privacy imposes a point-wise upper bound on the divergence between the distributions $\mathbb{P}(\boldsymbol{y} | \mathcal{D})$ and $\mathbb{P}(\boldsymbol{y} | \mathcal{D}')$, conditioned on the use of either one of two ``neighboring" global data sets $\mathcal{D}$ and $\mathcal{D}'$. The formal definition is given as follows. 

\begin{definition}[Differential Privacy (DP) \cite{dwork2006differential}] The learning process is $(\epsilon, \delta)$-differentially private if for any two possible adjacent global data sets $\mathcal{D}' = \cup_{k=1}^K \mathcal{D}'_k$ and $\mathcal{D}'' = \cup_{k=1}^K \mathcal{D}_k''$ where there exists only one client indexed $j$ satisfying $\|\mathcal{D}'_j - \mathcal{D}_j''\|_1 = 1$ and $\|\mathcal{D}'_k - \mathcal{D}_k''\|_1 = 0$ for any $k \neq j$, it holds that \begin{align}\label{eq:dp}
    \mathbb{P}(y | \mathcal{D}') \leq \exp(\epsilon) \mathbb{P}(y | \mathcal{D}'') + \delta,
\end{align}
where $\|\cdot \|_1$ is the Hamming distance \cite{waggener1995pulse}, $\epsilon \in (0, \infty)$ and $\delta \in [0, 1)$ are two constants characterizing the privacy budget. If \cref{eq:dp} holds, then the differential privacy loss, defined as 
\begin{align}\label{eq:dploss}
    \textstyle\mathcal{L}_{\mathcal{D}', \mathcal{D}''}(y) = \ln \frac{\mathbb{P}(y | \mathcal{D}')}{\mathbb{P}(y | \mathcal{D}'')}.
\end{align}
satisfies 
\begin{align}\label{new-eq:dploss}
     \mathbb{P}(\textstyle|\mathcal{L}_{\mathcal{D}', \mathcal{D}''}(y)| < \epsilon) >1-\delta. 
\end{align}
\end{definition}
Setting a sufficiently small $\epsilon$ and $\delta$ will forbid any adversary who knows all other data samples in the data set from identifying the remaining individual from the observed outputs.

\begin{lemma}[Privacy Loss Bound]\label{lm:privacy}
    The learning process of \emph{pAirZero}, equipped with analog modulation, is $(\epsilon, \delta)$-DP if \begin{align}
        \textstyle\sum_{t=1}^T\big(\sqrt{2} c^{(t)} \gamma^{(t)}/m^{(t)}\big)^2 \leq R_{\mathsf{dp}}(\epsilon, \delta), \ \forall k,
    \end{align}
    where \begin{align}
        R_{\mathsf{dp}}(\epsilon, \delta) = \big(\sqrt{\epsilon + [C^{-1}(1/\delta)]^2} - C^{-1}(1 / \delta)\big) ^ 2,
    \end{align}
    $C(x)$ is a function defined as $C(x)= \sqrt{\pi} x e^{x^2}$, and $C^{-1}$ is the inverse function of $C(x)$.
\end{lemma}

\begin{proof}
    See Appendix \ref{pf:lm:privacy}.
\end{proof}

{\color{black}
\begin{remark}[Privacy efficiency]
Unlike \cite{liu2020privacy}, we privatize the gradient projection rather than the gradient itself, which is a long vector. While privatizing on a long vector results in a completely random vector of the same size as the gradient, leading to significant computation overhead per iteration, privatizing on only a scalar is convenient for resource-constrained client devices.
\end{remark}
}

{\color{black}\begin{algorithm}[t]
\caption{The proposed \emph{pAirZero} and \emph{Sign-pAirZero}}\label{alg:cap}
\begin{algorithmic}[1]
\Ensure Trained model $\boldsymbol{w}^{(T)}$.
\State Initialize $\boldsymbol{w}^{(0)}$, clients pull $\boldsymbol{w}^{(0)}$

\For {$t = 1, \dots, T$}
\State clients estimate \cref{eq:projgrad} in parallel
\If {Analog}
\State Determine $c^{(t)}$ and $\sigma_k^{(t)}$ from Theorem \ref{thm:1};
\State client sends \cref{eq:tx_analog};
\State edge server receives by \cref{eq:rec_analog};
\Else
\State Determine $c^{(t)}$ and $\sigma_k^{(t)}$ from Theorem \ref{thm:2};
\State client sends \cref{eq:tx_digital};
\State edge server receives by \cref{eq:rec_digital};
\EndIf
\State edge server estimate $\hat{p}^{(t)}$ by \cref{eq:estp} and broadcast it to all clients;
\State Each client updates $\boldsymbol{w}^{(t+1)} \gets \boldsymbol{w}^{(t)} - \eta \hat{p}^{(t)} \boldsymbol{z}^{(t)}$;
\EndFor
\end{algorithmic}
\end{algorithm}}

We summarize the algorithm in pseudo-code, as shown in Algorithm \ref{alg:cap}.


\section{Convergence and Optimality Gap Analysis}\label{convergence-pairzero}

{\color{black} To verify the feasibility of the algorithm design given the vastly reduced communication and memory overhead with an efficient privacy mechanism, we conduct theoretical analyses on the convergence for the proposed \emph{pAirZero} and \emph{Sign-pAirZero}.}Regarding the loss function, we make the following widely adopted assumptions. 

\begin{assumption}[Gradient Lipschitz continuity, \cite{bottou2018optimization}]\label{ass:lsmooth}
The gradient of the global loss function $F(\boldsymbol{w})$ is assumed to globally $L$-Lipschitz continuous, namely, 
\begin{align}
    \|\nabla F(\boldsymbol{w}) - \nabla F(\boldsymbol{w}')\|_2^2 \leq L \|\boldsymbol{w} - \boldsymbol{w}'\|_2^2, \ \forall \boldsymbol{w}, \boldsymbol{w}' \in \mathbb{R}^d
\end{align}
\end{assumption}


\begin{assumption}[Polyak-\L ojaciewicz property, \cite{polyak1964gradient}]\label{ass:pl}
The global loss function $F(\boldsymbol{w})$ is said to have the Polyak-\L ojaciewicz property if it holds
\begin{align}
    \|\nabla F(\boldsymbol{w})\|_2^2 \geq 2 M (F(\boldsymbol{w}) - F^*), \ \forall \boldsymbol{w} \in \mathbb{R}^d,
\end{align}
where $F^*$ is the global optimum and $M \in (0, \infty)$ is a constant. 
\end{assumption}


\begin{assumption}[Sample-wise bounded gradient projection]\label{ass:boundgradproj} 
In iteration $t$, the gradient projection $p_k^{(t)}(\boldsymbol{u}, v)$ computed at any training sample $(\boldsymbol{u}, v)\in \mathcal{D}_k$ is assumed to be smaller than a constant $\gamma^{(t)}_k > 0$. For convenience, we also define $\gamma^{(t)}\triangleq \max{\gamma^{(t)}_k}$.
\end{assumption}

\begin{assumption}[Local $r$-effective rank, \cite{malladi2023fine}]\label{ass:r_effective_rank}
Define $G(\boldsymbol w^{(t)}) \triangleq \max_{(\boldsymbol{u}, v)\in \mathcal{D}} ||\nabla F(\boldsymbol w^{(t)};\boldsymbol{u}, v)||_2$. The global loss function is said to have local $r$-effective rank if, for any $\boldsymbol{w}^t$, there exists an effective rank-$r$ matrix $\boldsymbol{H}_{\boldsymbol w^{(t)}} \preccurlyeq L\cdot \boldsymbol{I}_d$ such that $\nabla^2 F(\boldsymbol w) \leq \boldsymbol{H}_{\boldsymbol w^{(t)}}$ holds for $\forall \boldsymbol{\boldsymbol w}\in\{\boldsymbol{w} \ | \ ||\boldsymbol w - \boldsymbol w^{(t)}||_2 \leq \eta \cdot d \cdot G(\boldsymbol w^{(t)})\}$, where the effective rank is defined as the value of $\text{tr} (\boldsymbol{H}_{\boldsymbol w^{(t)}})/\|\boldsymbol{H}_{\boldsymbol w^{(t)}}\|_{\text{op}}$. 
\end{assumption}

\begin{assumption}[Unbiased batch gradient with finite variance, \cite{bottou2018optimization}]\label{ass:finite_batch_variance}
The batch gradient is unbiased, and the variance of the FO batch gradient estimation is finite, specifically, 
\begin{align}
    \mathbb{E}[\nabla F(\boldsymbol{w}; \mathcal{B})] &= \nabla F(\boldsymbol{w}),\\
    \mathbb{E}[\|\nabla F(\boldsymbol{w};\mathcal{B})\|_2^2] &= \|\nabla F(\boldsymbol{w})\|_2^2 + \text{tr} (\boldsymbol{\Sigma}) /b,
\end{align}
where $b$ is the batch size of $\mathcal{B}$, $\boldsymbol{\Sigma}$ represents the covariance matrix of the true gradient $\nabla F$.
\end{assumption}

In the above, the gradient Lipschitz continuity as well as the Polyak-\L ojaciewicz property are widely assumed in federated learning, as they are key to proving the convergence of gradient-based optimization methods. The bounded gradient projection assumption is crucial for avoiding infinite privacy loss in a single iteration. At last, the local $r$-effective rank and finite batch gradient variance assumptions are widely used in the analysis ZO when fine-tuning an LLM. In particular, the local $r$-effective rank assumption ensures that the convergence speed of ZO is independent of the ambient dimension, which is essential; otherwise, the ZO gradient would be unable to achieve a sufficient decrease in the objective function of the LLM.

The global convergence behavior of \emph{pAirZero} with analog modulation is summarized as the following theorem: 

\begin{theorem}[Optimality Gap Bound of \emph{pAirZero}]\label{lm:convergence}
    Under \cref{ass:lsmooth,ass:pl,ass:r_effective_rank}, with a sufficiently small $\eta$ detailed in the proof, after $T$ iterations, the expected optimality gap is upper bounded as \begin{align}\label{eq:ogap_a}
        \textstyle \mathbb{E}[G^{(T)}] \leq A^T G^{(0)} + \sum_{t=1}^T \frac{\eta^2 L S O_r (m^{(t)})^2}{2 b A^{t - T} (K c^{(t)})^2},
    \end{align}
    where $S:= \max_t \text{tr}( \Sigma^{(t)})$, $G^{(t)} = F(\boldsymbol{w}^{t}) - F^*$ represents the optimality gap at iteration $t$, $O_r$ is the corrected low-rank factor of the model, $b$ is batch size, and $A = 1 - M \eta$ is the contraction factor.
\end{theorem}
\begin{proof}
    See Appendix~\ref{pf:lm:convergence}.
\end{proof}

In the above, since $A$ is smaller than $1$, $A^{T}G^{(0)}$ vanishes with a linear rate. Besides, the second term in the right-hand side of \cref{eq:ogap_a} is upper-bounded by a constant number (since it is a sum of a geometric decaying sequence whose common ratio is $A$). This theorem indicates that \emph{pAirZero} converges linearly to a neighborhood of the global optimal point, provided that $m^{(t)}$ and $c^{(t)}$ are upper bounded. In fact, the linear convergence rate is the best result that can be obtained with the PL condition. In the next section, we will also show how to minimize the neighborhood term by tuning the per-client transmit power and the strength of the artificial noise. 

Unlike analog modulation, where the desired signals are directly aggregated, digital modulation relies on a majority voting scheme to reach a consensus. This leads to a completely different convergence analysis, which is summarized in the following theorem: 


\begin{theorem}[Optimality Gap Bound for \emph{Sign-pAirZero}]\label{lm:convergence-digital}
Under \Cref{ass:lsmooth,ass:pl,ass:boundgradproj,ass:finite_batch_variance}, with a sufficiently small $\eta$ detailed in the proof, after $T$ iterations, the expected optimality gap is upper-bounded as
\begin{align}
    &\textstyle \mathbb{E}[G(\boldsymbol{w}^{(T)})]\leq \textstyle \tilde{A}^T G^{(0)} + \sum_{t=1}^T \tilde{A}^{T-t}\big(\theta \cdot (e^{(t)})^2 + r\big)
    \label{theorem-inequa}
\end{align}
where $\tilde{A} = 1 - \frac{M \eta^2}{\theta \pi}$ is the contraction factor, and $e^{(t)}$ is the total reversed sign probability at iteration $t$.
\end{theorem}
\begin{proof}
    See Appendix~\ref{proof:convergence-digital}.
\end{proof}

Different from the analog case, the quantity $e^{(t)}$ does not have an explicit expression. Nevertheless, we can derive an upper bound for $e^{(t)}$, as shown in the following lemma. 
\begin{lemma}[Upper Bound of $e^{(t)}$]\label{upper-bound-et}
Define $e^{(t)}$ as the sign reversing probability, namely, 
\begin{align}\label{eq:def_et}
    e^{(t)} \triangleq \mathbb{P}( \hat{p}^{(t)} \nabla F(\boldsymbol{w}^{(t)})^\top \boldsymbol{z}^{(t)} < 0 ).
\end{align}
Assuming that the probability of the sign of the batch gradient projection estimator differs from that of the true gradient projection is no larger than $e_0$, namely, 
\begin{align}\label{eq:e_0}
    e_k^{(t)} \triangleq \mathbb{P}(p_k^{(t)} \nabla F(\boldsymbol{w}^{(t)})^\top \boldsymbol{z}^{(t)} < 0) \leq e_0.
\end{align}
Also assume that $e_0\leq 1/2$. Then we have 
\begin{align}
    \textstyle (e^{(t)})^2 \leq \frac{4Ke_0(1-e_0) + \frac{(m^{(t)})^2}{(c^{(t)})^2}}{4Ke_0(1-e_0) + \frac{(m^{(t)})^2}{(c^{(t)})^2} + K^2(1-2e_0)^2}
    \label{eq:e-upper-bound}
\end{align}
provided that $0 < e_k^{(t)} < 1/2$.
\end{lemma}
\begin{proof}
See Appendix \ref{pf:upper-bound-et}.
\end{proof}

\begin{remark}
Although it is difficult to theoretically justify the validity of $e_0\leq 1/2$, numerical experiments indicate this is indeed the case, see \cref{srp}. 
\end{remark}

{\color{black} Noticing that for both cases, the global loss converges linearly to the optimal point with a bounded neighborhood, whose size is a summation on a geometric series with their common ratio no larger than $1$.}

With the derivation of the convergence bound, it turns out that the optimality gap bound is dominated by both the effective noise $m^{(t)}$ and channel gain $c^{(t)}$. {\color{black} However, the answer to the following question remains unclear: how can we achieve optimal performance by carefully setting the per-iteration transmission power and artificial noise scale while avoiding constraint violations?}

\section{{\color{black} Closing Optimality Gap via Power Control}}
\label{optmini}
{\color{black} Towards building a communication-memory-privacy-efficient system for LLM FFT, we incorporated artificial noise and channel noise }into the uploading process to facilitate DP protection {\color{black} on the gradient projection admitted by ZO optimization}. To this aim, based on the convergence bound provided in \cref{convergence-pairzero}, we build optimization problems for optimality gap minimization by tuning $c^{(t)}$ and $\sigma_k^{(t)}$ dynamically.

\subsection{\color{black} Power Allocation Optimization for \emph{pAirZero}}
Ignoring the constants, the optimality gap bound of \cref{eq:ogap_a} yields
\begin{align}
\textstyle \sum_{t=1}^T A^{-t} \big(\sum_{k=1}^K (\sigma_k^{(t)})^2 + \frac{N_0}{(c^{(t)})^2}\big)
\end{align}
The bound minimization problem with DP and per-client transmit power constraints is formulated as 
\begin{align}
    &\textstyle \min\limits_{\{c_k^{(t)}, \sigma_k^{(t)}\}} \sum_{t=1}^T A^{-t} \big(\sum_{k=1}^K (\sigma_k^{(t)})^2 + \frac{N_0}{(c^{(t)})^2}\big) \tag{P1}\label{eq:P1} \\ 
    &\textstyle\text{s.t.} \quad \sum_{t=1}^T \frac{2 (\gamma^{(t)})^2}{\sum_{k=1}^K (\sigma_k^{(t)})^2 + \frac{N_0}{(c^{(t)})^2}} \leq R_{\mathsf{dp}}(\epsilon, \delta), \tag{C1} \label{cons-C1-ana} \\
    &\textstyle {\color{white} \text{s.t.} \quad}\Big(\frac{c^{(t)}}{h_k^{(t)}}\Big)^2\left((\gamma_k^{(t)})^2 + d(\sigma_k^{(t)})^2\right) \leq P, \ \forall k, t. \tag{C2}\label{cons-power-ana}
\end{align}
where the first constraint is equivalent to the DP constraint, and the second constraint is the power constraint. Recall that the original power constaint is $\mathbb{E}\big[\|x_k^{(t)}\|^2\big] \leq P$, and (\ref{cons-power-ana}) is obtained by invoking $x_k^{(t)}=\alpha_k^{(t)} (p_k^{(t)} + n_k^{(t)})$. The following lemma explains the equivalence between (\ref{cons-C1-ana}) and the DP constraint.

Despite the complex form of the problem (\ref{eq:P1}), we show that the optimal solution to this problem can be obtained in closed form. As such, the per-client transmit power and the strength of the artificial noise can be conveniently tuned in each iteration. 

\begin{theorem}[{\color{black} Closed-form Solution to \ref{eq:P1}}]\label{thm:1}
    The optimal solution to \eqref{eq:P1} is given as follows:
    \begin{itemize}
        \item If condition \begin{align}
           \textstyle \frac{2P}{K N_0} \sum_{t=1}^T \min_k \{(h_k^{(t)})^2\} < R_{\mathsf{dp}}(\epsilon, \delta)
        \end{align}
        holds true, then \begin{align}
            c^{(t)^*} = \min_k \{P^{\frac12} h_k^{(t)} (\gamma_k^{(t)})^{-1}\}, \sigma_k^{(t)^*} = 0,
        \end{align}
        is the unique solution to \eqref{eq:P1}, which means that the client with the poorest channel condition transmits the uncoded projection with full power $P$.
        \item Otherwise, \eqref{eq:P1} admits non-unique solutions, among which the solution that minimizes the transmit power among all clients is given as 
        \begin{align}\label{eq:solutiona}
          \textstyle  c^{(t)^*} = \min \Big\{ \underbrace{\textstyle\frac{A^{- \frac{t}{4}} N_0^{\frac12}}{(2 \zeta^*)^{\frac14} (\gamma^{(t)})^{\frac12}}}_{\text{adaptive term}} , \min_k \big\{\frac{P^{\frac12} h_k^{(t)}}{\gamma_k^{(t)}}\big\}\Big\}, \sigma_k^{(t)^*} = 0.
        \end{align}
    \end{itemize}
    where the value of $\zeta^{*}$ can be obtained by bisection search so as to meet the DP constraint, namely, 
    \begin{align}\label{eq25}
    \textstyle  \sum_{t=1}^T 2(\gamma^{(t)})^2 \min \Big\{\frac{A^{-\frac{t}{2}}}{(2 \zeta)^{\frac12}\gamma^{(t)}},  \min_k\big\{\frac{P (h_k^{(t)})^2}{N_0(\gamma_k^{(t)})^2}\big\}\Big\} = R_{\mathsf{dp}}(\epsilon, \delta).
\end{align}

\end{theorem}

\begin{proof}
    The proof of this theorem is very similar to that of Theorem \ref{eq:P2}, therefore, we omit it. 
\end{proof}

\subsection{\color{black} Power Allocation Optimization for \emph{Sign-pAirZero}}
Similar to the analog case, by ignoring the constants, the right-hand side of \cref{theorem-inequa}, we can construct the following optimality gap bound minimization problem: 
\begin{align}
    &\textstyle\min\limits\nolimits_{\{c_k^{(t)}, \sigma_k^{(t)}\}} \ \sum_{t=1}^T \tilde{A}^{-t} \eta \theta (e^{(t)})^2 \label{eq:P2}\tag{P2} \\ 
    &\textstyle \text{s.t.} \quad \sum_{t=1}^T \frac{2}{\sum_{k=1}^K (\sigma_k^{(t)})^2 + \frac{N_0}{(c^{(t)})^2}} \leq R_{\mathsf{dp}}(\epsilon, \delta), \tag{C3}\label{eq:C3} \\
    &{\color{white}\textstyle \text{s.t.} \quad}\big(c^{(t)}/h_k^{(t)}\big)^2(1 + d(\sigma_k^{(t)})^2) \leq P, \quad\forall k, t. \tag{C4}\label{eq:C4}
\end{align}

In the above problem, the equivalence between the constraint (\ref{eq:C3}) and the DP-constraint is shown in the following lemma, and the power constraint follows immediately by realizing that $\gamma = 1$.
\begin{lemma}[Privacy Loss Bound for \emph{Sign-pAirZero}]\label{lm:privacy-digital}
The learning process of \emph{pAirZero}, equipped with digital modulation, is \((\epsilon, \delta)\)-DP if, for each device $k$,
\begin{equation}
\textstyle \sum_{t=1}^{T} \frac{2(c^{(t)})^2}{ (c^{(t)})^2 \sum_{k=1}^K (\sigma_k^{(t)})^2 + N_0 } \le R_{\mathsf{dp}}(\epsilon, \delta)
\end{equation}
\end{lemma}
\begin{proof}
  Since clients send only a sign, setting $\gamma = 1$ in \cref{lm:privacy} yields the desired result.
\end{proof}

At last, we can also derive a closed-form solution to problem (\ref{eq:P2}), see the following theorem. 
\begin{theorem}[{\color{black} Closed-form Solution to \ref{eq:P2}}]
\label{thm:2}
The optimal solution to problem (\ref{eq:P2}) is given as follows:
    \begin{itemize}
        \item If condition
\begin{align}
\textstyle \min_k\sum_{t=1}^T     \frac{2(h_k^{(t)})^2P}{KN_0} \le R_{\mathsf{dp}}(\epsilon, \delta)
\end{align}
holds, then 
\begin{align}
            c^{(t)^*} = \min_k \{P^{\frac12} h_k^{(t)}\}, \ \sigma_k^{(t)^*} = 0,
\end{align}
is the unique optimal solution to problem (\ref{eq:P2}), which means that the client with the poorest channel condition transmits the uncoded gradient projection with full power $P$.
        \item Otherwise, problem (\ref{eq:P2}) admits non-unique optimal solutions, among which the solution that minimizes the transmit power among all clients is given as  \begin{flalign}\label{eq:solutiond}
    & c^{(t)^*} =\textstyle \min \Big\{ \min_k \{P^{\frac12} h_k^{(t)}\}, & \notag\\
    &\textstyle\quad \underbrace{\textstyle \frac{N_0^\frac{1}{2}\left(2 \left( \tilde{A}^{-t}B_2^2 - 2\zeta^* \right)\right)^\frac{1}{2}}{\left(\left( B_1 + B_2 \right) \left( 4\zeta^* + \sqrt{8 \tilde{A}^{-t} B_2^2 \zeta^*} \right)\right)^\frac{1}{2}}}_{\text{adaptive term}}\Big\}, \sigma_k^{(t)^*} = 0.&
\end{flalign}
where the value of $\zeta^*$ can be obtained by bisection so as to satisfy the DP constraint, namely, \begin{align}
        &\textstyle \sum_{t=1}^T 2 \min \Big\{\min_k\{P (h_k^{(t)})^2\},\notag\\
&\textstyle \frac{2 \left( \tilde{A}^{-t}B_2^2 - 2\zeta \right)}{\left( B_1 + B_2 \right) \left( 4\zeta + \sqrt{8 \tilde{A}^{-t} B_2^2 \zeta} \right)}\Big\} =R_{\mathsf{dp}}(\epsilon, \delta).
    \end{align}
\end{itemize}
\end{theorem}

\begin{proof}
    See \cref{pf:P3}.
\end{proof}

\section{Simulations}\label{sec:sim}
In this section, we provide simulation results to demonstrate the superiority of the proposed \emph{pAirZero} and \emph{Sign-pAirZero}. All the experiments are conducted on the OPT-125M model. In particular, we demonstrate that both \emph{pAirZero} and \emph{Sign-pAirZero} achieve comparable performance to non-DP cases and consistently outperform the baselines. All empirical results were obtained on a machine equipped with an AMD EPYC 7742 64-Core Processor and four NVIDIA A100-SXM4-80GB GPUs.

\subsection{Experimental Settings}




\begin{table}[t]
\centering
\caption{Learning rate grid search, selected learning rates are \textbf{bolded}.}\label{table:lr_grid}
\begin{tabular}{@{}rc@{}}
\toprule
        & $\eta$                             \\ \midrule
\emph{pAirZero}  & $\{1e-7, 3e-7, \mathbf{5e-7}, 1e-6, 5e-6\}$ \\
\emph{Sign-pAirZero} & $\{1e-6, 5e-6, 1e-5, \mathbf{3e-5}, 5e-5\}$ \\ \bottomrule
\end{tabular}
\end{table}

We employ OPT-125M \cite{zhang2022opt}, a well-recognized light-weight LLM. The number of clients is assumed to be $K=5$. We include two classic language tasks for comparison: 1. SST-2\cite{socher2013recursive} (Stanford Sentiment Treebank, binary version), a binary sentence-level sentiment classification task (positive or negative); 2. SQuAD\cite{rajpurkar2016squad} (Stanford Question Answering Dataset), a reading comprehension task where the model is asked to extract the correct answer span directly from the provided passage. We select these two tasks to comprehensively evaluate the model's improvement during fine-tuning, considering both basic abilities, such as understanding and classification (SST2), and more advanced abilities, including contextual understanding and information retrieval (SQuAD). In terms of parameter setting, the perturbation scale in \cref{eq:projgrad} is set to $\mu = 0.001$, and the number of training samples is fixed as $1000$. Moreover, we also set $\epsilon=5$, $\delta=0.01$, $T=8000$ throughout this section. For each reported instance, we conduct a grid search to determine the learning rate that yields the best performance. The search grid is shown in \Cref{table:lr_grid}. Each data point in the main results represents the average of $4$ trials, where we perform four independent runs with random seeds and report the mean and standard deviation of the corresponding measurements to minimize bias due to chance. As for the value of $e_0$ and $A $, and $\gamma$, we set $e_0 = 0.4960$, $A = \tilde{A} = 0.998$, and $\gamma=100$, these choices will be elaborated in \cref{srp}. 

To evaluate the performance w.r.t. different channel conditions, we define the maximum signal-to-noise ratio as \begin{align}\label{eq:maxsnr}
    \textstyle \text{SNR}_{\max} = P / dN_0,
\end{align} where $dN_0$ represents the power of the channel noises within one communication block, so that devices may optimally transmit with a power strictly smaller than $P$. 


\subsection{Main Results}

We implement the proposed \emph{pAirZero} and \emph{Sign-pAirZero} across wireless channels with different maximum SNRs. We compare the solution-induced power allocation with the following cases:
\begin{enumerate}
    \item \texttt{Perfect}. It provides the upper bound of the metrics in evaluation. In this sense, we assume that the aggregation is noise-free, that is, 
    \begin{align}\label{eq:globalg}
    \textbf{\emph{pAirZero:}}&\quad\textstyle \boldsymbol{g}^{(t)} = \frac{1}{K}\sum_{k=1}^K p_k^{(t)} \boldsymbol{z}^{(t)},
    \\
    \label{eq:globalgd}
    \textbf{\emph{Sign-pAirZero:}}& \quad \textstyle\boldsymbol{g}^{(t)} = \text{sign}\big\{\sum_{k=1}^K \text{Sign}\{p_k^{(t)}\}\big\} \boldsymbol{z}^{(t)}.
    \end{align}
    \item \texttt{Static}. Instead of featuring an adaptive power allocation, it distributes privacy budget evenly across all training iterations by replacing the adaptive term in $c^{(t)}$ in \cref{eq:solutiona} and \cref{eq:solutiond} with the following constant: \begin{align}\label{eq:constantc}
        c^{(t)}=\textstyle \sqrt{\frac{N_0 R_{\sf dp}(\epsilon, \delta)}{2 T \gamma^2}}.
    \end{align}
    \item \texttt{Reversed}. It is observed that the adaptive term of the optimization-induced solutions in both the \emph{pAirZero} and \emph{Sign-pAirZero} implies an increasing channel gain $c^{(t)}$. To verify the effectiveness of this increasing trend, we replace $A^{-\frac{t}{4}}$ and $\tilde{A}^{-t}$ by $A^{\frac{t}{4}}$ and $\tilde{A}^{t}$, respectively. 
\end{enumerate}



\begin{figure}[t]
    \centering
    \includegraphics[width=1.0\linewidth]{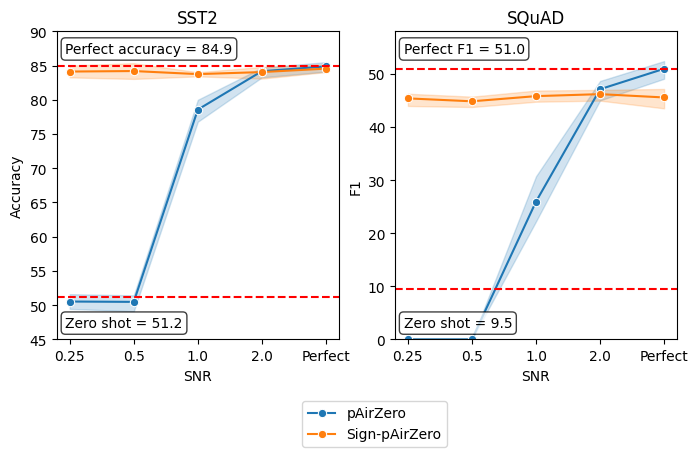}
    \caption{Main results on OPT-125M with SST2 and SQuAD task.\label{fig:big exp}}
\end{figure}

\begin{figure}[t]
    \centering
    \includegraphics[width=1.0\linewidth]{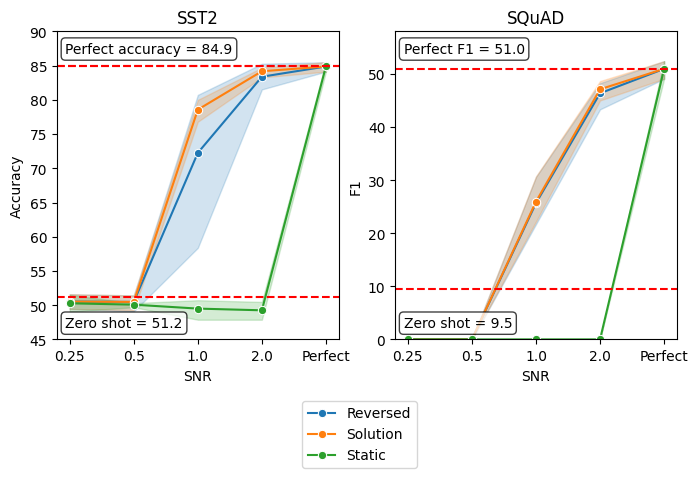}
    \caption{Performance with and without solution-based power allocation with analog modulation.\label{fig:ablation}}
\end{figure}

\textbf{\emph{Sign-pAirZero} is more stable across SNR with privacy constraints.} We report the main results in \cref{fig:big exp}. We found that digital OTAs exhibit small fluctuations in performance (within one standard deviation), whereas analog OTAs are completely compromised in the low SNR regime. This is due to the zeroth-order gradient projection having a wide spread, detailed in \Cref{fig:pdist}. This is because the DP model assumes full privacy loss, with the fluctuation of the disclosed function always in its maximum, as in \Cref{lm:privacy}. However, the privacy loss is often much smaller since the fluctuation is usually smaller. In other words, analog OTA requires a privacy model to overestimate privacy loss, leading to more noisy updates with lower SNR.

\textbf{\emph{pAirZero} has higher possible performance.} We observe that in harder tasks like SQuAD, analog OTA has a lead for more than one standard deviation compared to its digital version. This is due to the larger gradient noise introduced by the one-bit compression on the gradient projection.

\textbf{Solution-based power allocation outperforms other baselines.} We ablate on the power allocation scheme and report the results in \cref{fig:ablation}. It is shown that \verb|Solution| outperforms other baselines. Moreover, it is found that the model is compromised under the \verb|Static| case. This is because ZO-based fine-tuning requires a large number of aggregations, which necessitates a minimal channel gain \cref{eq:constantc} with a large value of $T$, highlighting the need for optimization in \cref{eq:P1}. 

{\color{black}

\subsection{Memory and Communication Efficiency}\label{memory-communication-efficiency}

\begin{table}[h]
\caption{Estimated local minimum memory overhead and per-iteration upload for fine-tuning OPT-125M.}
\centering
\label{table:memory_cost}
\begin{tabular}{lrr}
\hline
              & Memory cost   & Per-iteration Upload \\ \hline
Model size    & $238.88$ MB   & -                    \\ \hline
\emph{Sign-pAirZero} & $\sim 250$ MB & $1$ bit              \\
\emph{pAirZero}      & $\sim 250$ MB & $16$ bits            \\
FO SGD        & $\sim 600$ MB & $238.88$ MB          \\
FO Adam       & $955.58$ MB & $238.88$ MB          \\ \hline
\end{tabular}
\end{table}

We present a minimum memory overhead and per-iteration upload analysis in \cref{table:memory_cost} assuming an FP16 precision. We observe a magnitude-of-orders reduction in per-iteration upload and $75\%$ less memory cost in \emph{Sign-pAirZero} and \emph{pAirZero}, compared to conventional methods relying on FO methods. The massive reduction in memory cost and upload will make \emph{pAirZero} and \emph{Sign-pAirZero} preferred for resource-constrained clients in an FL system.

}

\subsection{Parameter Study}\label{srp}
Since the analysis and problem solution involve key numerical parameters that cannot be directly obtained (e.g., $e_0$, $A$, and $\tilde{A}$) or need to be determined manually (e.g., $\gamma$), we therefore conduct a study to justify our choices of their value.

\subsubsection{Sign Reversing Probability $e_0$}

As stated in \Cref{lm:convergence-digital}, the convergence of \emph{Sign-pAirZero} is dependent on the sign-reversing probability $e^{(t)}$, so knowledge of the range of this value is necessary for us to understand the convergence behavior of \emph{Sign-pAirZero}. Since the explicit expression of $e^{(t)}$ is non-obtainable due to the complicated dynamics from data batches to gradient projections, we resort to an empirical study.

\begin{figure}[h]
    \centering
    \includegraphics[width=0.5\linewidth]{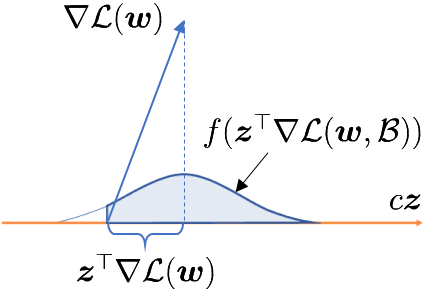}
    \caption{Inherent sign-reversing probability density.\label{fig:pte}}
\end{figure}





\textbf{Simulation Settings.}
We run OPT-125M on the SST2 task. We choose a training set consisting of $5000$ samples, and for every $4000$ iterations, we evaluate the gradient projection by sampling the gradient direction corresponding to seeds $s = 0$ to $39$. We average the gradient projections to obtain $\boldsymbol{z}_s^\top \mathcal{L}(\boldsymbol{w}^{(t)})$. We then uniformly sample $10000$ batches of size $64$ and take an average of the gradient projection $\boldsymbol{z}_s^\top \mathcal{L}(\boldsymbol{w}^{(t)}, \mathcal{B})$. We compute $e_k^{(t)}$ as the proportion of the batches holding a batch gradient projection with its sign different from that of $\boldsymbol{z}_s^\top\nabla \mathcal{L}(\boldsymbol{w}^{(t)})$.

\textbf{Range of Inherent Sign-Reversing Probability.}
We report the measured $e_k^{(t)}$ in \Cref{fig:pte_sim1}. It is noticed that the gradient projections are generally small, mainly due to the high dimensionality of the ambient space. The highest reading of $e_k^{(t)}$, namely $0.4968$, is obtain at $t = 12000$, with $\boldsymbol{z}^\top \nabla \mathcal{L}(\boldsymbol{w}) = -0.3330$.

\begin{figure*}[h]
    \centering
    \includegraphics[width=0.95\linewidth]{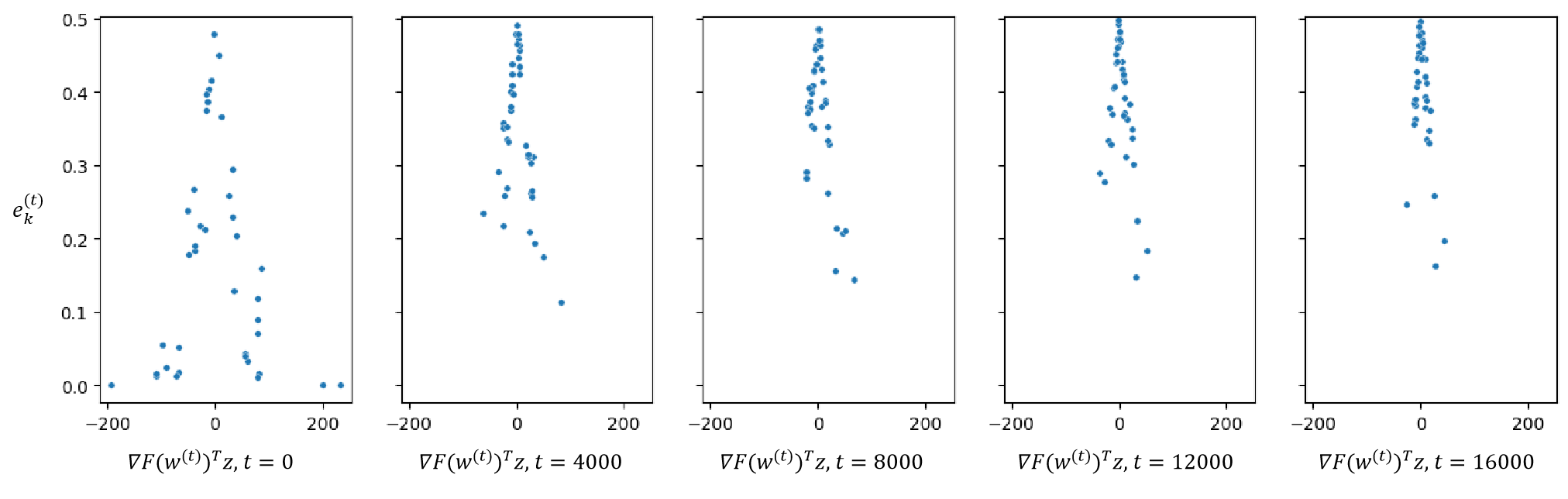}
    \caption{Inherent sign-reversing probability vs. $\boldsymbol{z}_s^\top\nabla\mathcal{F}(\boldsymbol{w}^{(t)})$.}
    \label{fig:pte_sim1}
\end{figure*}

\textbf{Near-Symmetric Distribution of $\boldsymbol{z}^\top \nabla\mathcal{L}(\boldsymbol{w}, \mathcal{B})$.}
In \Cref{fig:pte_sim2}, we report the value of $\boldsymbol{z}_s^\top \nabla\mathcal{L}(\boldsymbol{w}, \mathcal{B})$ from $s=0$ to $4$, with $t=0, 4000, 8000, 12000, 16000$. The red lines marks the corresponding $\boldsymbol{z}^\top \nabla \mathcal{L}(\boldsymbol{w})$. The distributions exhibit an obvious symmetric pattern.

In summary, upon modeling the distribution of the batch gradient projections as a symmetric distribution centered at the proper gradient projection, the local sign-reversing probability is always smaller than $1/2$. Also, we can observe that the analog sensitivity function overestimates the privacy loss mainly due to the wide distribution of the batch gradient projection.

\begin{figure}[h]
    \centering
    \includegraphics[width=1\linewidth]{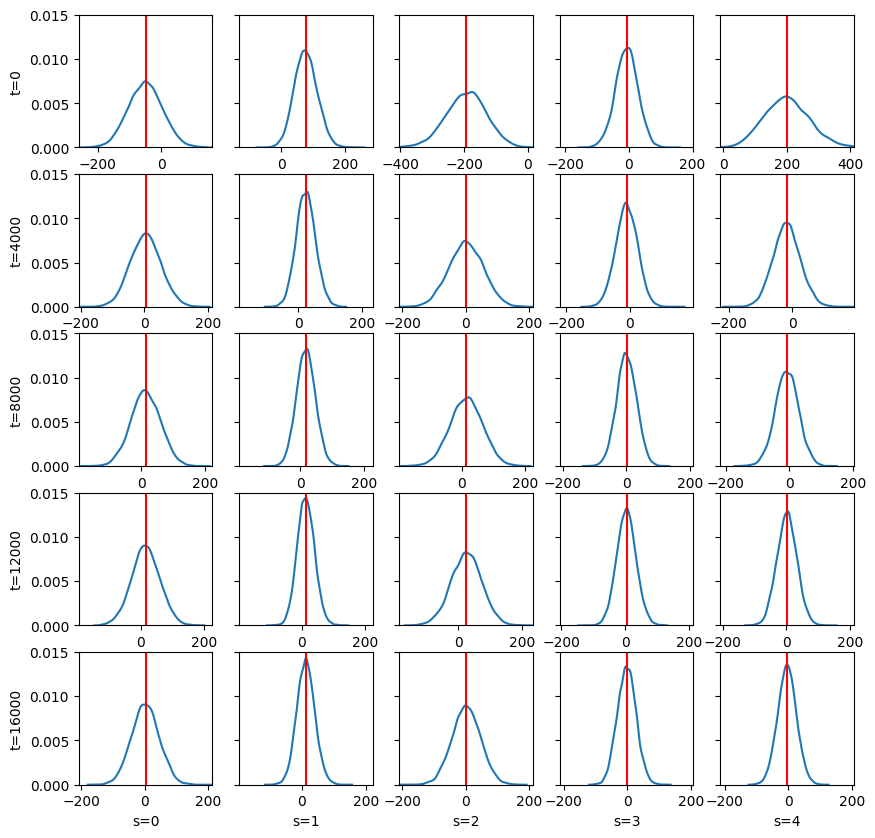}
    \caption{Distribution of $\boldsymbol{z}_s^\top\nabla\mathcal{L}(\boldsymbol{w}, \mathcal{B})$. The red lines are estimated $\boldsymbol{z}^\top \nabla\mathcal{L}(\boldsymbol{w})$.}
    \label{fig:pte_sim2}
\end{figure}

\subsubsection{Contraction Ratio $A$ and $\tilde{A}$}
The contraction ratio absorbs numerous intractable parameters such as $L$, $M$, and $O_r$. Nevertheless, this quantity plays a critical role in the power control. Therefore, an empirical estimation of this parameter is necessary. To this end, we record the loss values throughout one fine-tuning process on OPT-125M with the SST2 task, and estimate the upper bound of $A$ as $\min_t \{ (G^{(t)} / G^{(0)}) ^{1 / t} \}$. The estimated upper bound is $0.998$. We hence set $A$ and $\tilde{A}$ as this value.

\subsubsection{Gradient Projection Clip Threshold $\gamma$}

To avoid infinite privacy loss, we assert \Cref{ass:boundgradproj}. However, gradient projections have a wide range of applications. An overly large $\gamma$ will accordingly result in an overestimated privacy loss, while an excessively small $\gamma$ leads to frequent gradient clips. Both of these situations are undesirable because they all jeopardize the convergence speed of the proposed algorithm. To get a favorable value of $\gamma$, we recorded gradient projections throughout one fine-tuning process on OPT-125M with the SST2 task and made a histogram, which is shown in \Cref{fig:pdist}. It is found that over $97\%$ of the gradient projections fall within $[-100, 100]$, hence we set $\gamma = 100$.

\begin{figure}[h]
    \centering
    \includegraphics[width=0.8\linewidth]{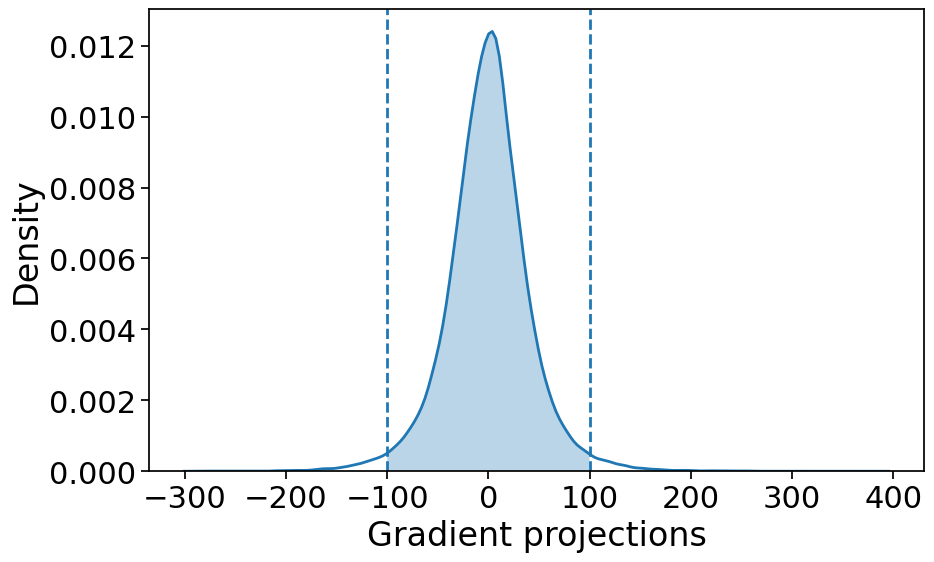}
    \caption{Distribution of gradient projections.\label{fig:pdist}}
\end{figure}



\section{Concluding Remarks}\label{sec:conclude}
In this paper, we have proposed a communication-memory-privacy-efficient over-the-air (OTA) transmission-based federated fine-tuning method, termed \emph{pAirZero}, along with its digital adaptation, \emph{Sign-pAirZero}. The differential privacy (DP) mechanism is embedded into the gradient transmission process to prevent privacy leakage. The proposed method requires only bit-level communication loads and inference-level memory usage. Additionally, it removes the strict synchronization requirements, which are the major obstacle to conventional OTA-based methods. More importantly, we have formulated an optimization model to determine the optimal transmit power and artificial noise level per iteration, ensuring a guaranteed level of privacy protection regardless of the channel noise strength. Numerical experiments demonstrate the superiority of our method.




\appendix

\subsection{Proof of \cref{lm:privacy}}
\label{pf:lm:privacy}
To start, let $\boldsymbol{y}_k\triangleq[y_k^{(1)}, \dots, y_k^{(T)}]$ represent the received signal of client $k$ over all $T$ iterations. The privacy loss for client $k$ after $T$ iterations is given as
\begin{align}
     &\mathcal{L}_{\mathcal{D}', \mathcal{D}''}(\boldsymbol{y}_k)\textstyle = \log \Big(\prod_{t=1}^T \frac{P(y_k^{(t)} | y_k^{(t-1)}, \dots, y_k^{(1)}, \mathcal{D}')}{P(y_k^{(t)} | y_k^{(t-1)}, \dots, y_k^{(1)}, \mathcal{D}'')}\Big) \notag \\
    &\textstyle= \sum\limits_{t=1}^T \log\Big(\exp \big(\frac{-(y_k^{(t)} - c^{(t)} p_{k, \mathcal{D}'}^{(t)})^2}{2 (m_k^{(t)})^2}\big)/\exp \big(\frac{-(y_k^{(t)} - c^{(t)} p_{k, \mathcal{D}''}^{(t)})^2}{2 (m_k^{(t)})^2}\big) \Big) \notag \\
    &\textstyle= \sum\limits_{t=1}^T \log \Big(\exp\big(- \frac{(r_k^{(t)})^2}{2 (m_k^{(t)})^2}\big)/\exp\big(- \frac{(r_k^{(t)} + v_k^{(t)})^2}{2 (m_k^{(t)})^2}\big)\Big)
\end{align}
where $p_{k, \mathcal{D}'}^{(t)}$ is the local gradient projection at iteration $t$ given a dataset $\mathcal{D}'$, $r_k^{(t)}$ is the effective noise, the second line is obtained from \cref{eq:dploss}, the third line is obtained by realizing that $v_k^{(t)}$ is the difference of the observation with different local datasets, namely, 
\begin{align}
    v_k^{(t)} = h_k^{(t)} \alpha_k^{(t)} (p_{k, \mathcal{D}'}^{(t)} - p_{k, \mathcal{D}''}^{(t)}) \leq 2 c^{(t)} \gamma^{(t)}. 
\end{align} 
Following Appendix A in \cite{liu2020privacy}, the privacy violation probability can be bounded by 
\begin{align}
    &\textstyle \mathbb{P}\big( |\sum_{t=1}^T (2 r_k^{(t)} v_k^{(t)} + (v_k^{(t)})^2) / 2 (m_k^{(t)})^2 | > \epsilon\big) \notag \\
    \overset{(a)}{\leq}&\textstyle \mathbb{P}\big(|\sum_{t=1}^T r_k^{(t)}v_k^{(t)}/(m_k^{(t)})^2| > \epsilon - (\sum_{t=1}^T (v_k^{(t)})^2/2 (m_k^{(t)})^2)\big) \notag \\
    =& \textstyle 2 \mathbb{P}\Big( \big(\sum\limits_{t=1}^T r_k^{(t)} v_k^{(t)}/(m_k^{(t)})^2\big)> \epsilon - \big(\sum\limits_{t=1}^T (v_k^{(t)})^2/2 (m_k^{(t)})^2\big)\Big) \notag \\
    \overset{(b)}{\leq}&\textstyle \frac{2\sigma}{s \sqrt{2 \pi}} \exp \left(- \frac{s^2}{2 \sigma^2}\right), \label{eq:target}
\end{align}
where 
\begin{align}
    \textstyle \sigma = (\sum_{t=1}^T (v_k^{(t)}/m_k^{(t)})^2)^{1/2}, \quad s = \epsilon - \sum_{t=1}^T \frac12 (v_k^{(t)}/m_k^{(t)}), \notag
\end{align}
$(a)$ is because $\mathbb{P}(X < -\epsilon - b) \leq \mathbb{P}(X < -\epsilon + b)$ holds for an arbitrary $b \geq 0$, and $(b)$ comes from the following Mills' bound of tail probability of Gaussian probability $X \sim \mathcal{N}(0, \sigma^2)$: 
\begin{align}
    \textstyle\mathbb{P}(X > s) =&\textstyle \frac{1}{\sigma \sqrt{2 \pi}} \int_s^\infty \exp\left(-\frac{x^2}{2 \sigma^2}\right) \textrm{d} x \notag \\
    \leq& \textstyle\frac{1}{\sigma \sqrt{2 \pi}} \int_s^\infty \frac{x}{s} \exp\left(-\frac{x^2}{2 \sigma^2}\right) \textrm{d}x \notag \\
    =& \textstyle\frac{\sigma}{s \sqrt{2 \pi}}\exp{\left(- \frac{s^2}{2 \sigma^2}\right)}. 
    \label{mill-bound}
\end{align}
From \cref{eq:target}, we can immediately obtain the desired result: 
\begin{align}
    \textstyle \mathbb{P}(|\mathcal{L}_{\mathcal{D}', \mathcal{D}''}(y_k)| > \epsilon) \leq \frac{1}{q \sqrt{\pi}}e^{-q^2} < \delta, \notag
\end{align}
where
\begin{align}
    \textstyle q = \frac{\epsilon - \sum_{t=1}^T \frac{1}{2}(v_k^{(t)} / m_k^{(t)})^2}{\sqrt{2 \sum_{t=1}^T (v_k^{(t)} / m_k^{(t)})^2}}, \quad C(x) = \sqrt{\pi} x e^{x^2}. \notag
\end{align}

\subsection{Proof of \cref{lm:convergence}}\label{pf:lm:convergence}

We denote the true zeroth-order batch gradient from mini-batch $\mathcal{B}$ and model $\boldsymbol{w}^{(t)}$ as $\hat{\nabla}F(\boldsymbol{w}^{(t)}; \mathcal{B})$. According to \cref{ass:lsmooth}, we have 
\begin{align}
    \textstyle F(\boldsymbol{w}^{(t+1)}) \leq& \textstyle F(\boldsymbol{w}^{(t)}) - \eta \nabla F(\boldsymbol{w}^{(t)}) ^\top \boldsymbol{g}^{(t)} + \frac{L \eta^2}{2}\|\boldsymbol{g}^{(t)}\|_2^2 \nonumber 
    \\
    \leq& \textstyle F(\boldsymbol{w}^{(t)}) - \eta \nabla F(\boldsymbol{w}^{(t)}) ^\top 
    \hat{\nabla}F(\boldsymbol{w}^{(t)}; \mathcal{B}) \nonumber
    \\
    & \textstyle + \frac{L \eta^2}{2} (1 + \frac{(m^{(t))^2}}{(K c^{(t)})^2})\|\hat{\nabla}F(\boldsymbol{w}^{(t)}; \mathcal{B})\|_2^2 \label{eq:second_line}
    \\
    \leq& \textstyle F(\boldsymbol{w}^{(t)}) - \eta \nabla F(\boldsymbol{w}^{(t)}) ^\top \nabla F(\boldsymbol{w}^{(t)}; \mathcal{B}) \label{eq:fourth_line}
    \\
    &\textstyle + \frac{L \eta^2 O_r}{2} (1 + \frac{(m^{(t))^2}}{(K c^{(t)})^2}) \|\nabla F(\boldsymbol{w}^{(t)};\mathcal{B})\|_2^2, \label{eq:fifth_line} 
\end{align}
where \cref{eq:second_line} comes from the fact that $\boldsymbol{g}^{(t)} = \hat{p}^{(t)} \boldsymbol{z}^{(t)}$ with $\hat{p}^{(t)} \sim \mathcal{N}(p^{(t)}, \frac{(m^{(t))^2}}{(K c^{(t)})^2})$, \cref{eq:fourth_line} is from the unbiasedness of $\hat{\nabla}F(\boldsymbol{w}^{(t)}; \mathcal{B})$, and \cref{eq:fifth_line} comes from proof of Theorem 1 (page 33) of \cite{malladi2023fine} (the norm of the ZO gradient scaled the quadratic norm by $O_r$ times compared to that of FO, where $O_r = \frac{dr + 2d}{d + 2}$ as we set the number of perturbation samples $n=1$ under \Cref{ass:r_effective_rank}). It then follows by \cref{ass:finite_batch_variance}:
\begin{align}
    &F(\boldsymbol{w}^{(t+1)})\leq  F(\boldsymbol{w}^{(t)}) - \eta \|\nabla F(\boldsymbol{w}^{(t)})\|_2^2\nonumber \\
    &\textstyle + \frac{L \eta^2 O_r}{2} (1 + \frac{(m^{(t)})^2}{(Kc^{(t)})^2}) (\|\nabla F(\boldsymbol{w}^{(t)})\|_2^2 + \frac{\text{tr} (\boldsymbol{\Sigma^{(t)}})}{b})\nonumber \\
    \leq&\textstyle F(\boldsymbol{w}^{(t)}) - (\eta - \frac{L \eta^2 O_r}{2}(1 + \frac{(m^{(t)})^2}{(Kc^{(t)})^2}))\|\nabla F(\boldsymbol{w}^{(t)})\|_2^2 \nonumber \\
    &+ \textstyle \frac{\eta^2 L O_r S (m^{(t)})^2}{2 b (K c^{(t)})^2} \nonumber \\
    \leq&\textstyle F(\boldsymbol{w}^{(t)}) - \frac{\eta}{2}\|\nabla F(\boldsymbol{w}^{(t)})\|_2^2 + \frac{\eta^2 L O_r S (m^{(t)})^2}{2 b (K c^{(t)})^2}, \label{eq:converg}
\end{align}
where \cref{eq:converg} is obtained by choosing $0 < \eta \leq \min_t\{1 / (LO_r(1 + \frac{(m^{(t)})^2}{(Kc^{(t)})^2}))\}$ and $S = \max_t \text{tr}(\Sigma^{(t)})$. We then subtract $F^*$ from both sides as follows:
\begin{align}
    \textstyle F(\boldsymbol{w}^{(t+1)}) - F^* \leq (1 - M \eta) (F(\boldsymbol{w}^{(t)}) - F^*) +  \frac{\eta^2 L O_r S (m^{(t)})^2}{2 b (K c^{(t)})^2}. \nonumber
\end{align}
At last, we obtain the desired result by telescoping the above inequality. 



\subsection{Proof of \cref{lm:convergence-digital}}\label{proof:convergence-digital}
By Taylor's expansion, we have
\begin{align*}
F(\boldsymbol{w}^{(t+1)}) &\leq F(\boldsymbol{w}^{(t)}) - \eta \nabla F(\boldsymbol{w}^{(t)})^\top (\boldsymbol{w}^{(t+1)} - \boldsymbol{w}^{(t)})\\
&\textstyle + \frac{1}{2}\eta^2(\boldsymbol{w}^{(t+1)} - \boldsymbol{w}^{(t)})^\top \boldsymbol{H}_{\boldsymbol{w}^{(t)}}(\boldsymbol{w}^{(t+1)} - \boldsymbol{w}^{(t)})
\end{align*}
Taking expectation over $\mathcal{B}$ and $\boldsymbol{z}$, we have
\begin{align*}
&\textstyle \mathbb{E}[F(\boldsymbol{w}^{(t+1)})]\\
\leq&\textstyle F(\boldsymbol{w}^{(t)}) - \eta \mathbb{E}_{z,\mathcal{B}}\big[\frac{\boldsymbol{z}^\top \nabla F(\boldsymbol{w}^{(t)}) \boldsymbol{z}^\top \nabla F(\boldsymbol{w}^{(t)}; \mathcal{B})}{|\boldsymbol{z}^\top \nabla F(\boldsymbol{w}^{(t)};\mathcal{B})|}\big]\\
&\textstyle + \mathbb{E}_{z,\mathcal{B}}\big[\frac{\eta^2\boldsymbol{z}^\top \nabla F(\boldsymbol{w}^{(t)};\mathcal{B}) \boldsymbol{z}^\top \boldsymbol{H}_{\boldsymbol{w}^{(t)}}\boldsymbol{z}^\top\nabla F(\boldsymbol{w}^{(t)}; \mathcal{B})z}{2(\boldsymbol{z}^\top \nabla F(\boldsymbol{w}^{(t)};\mathcal{B}))^2} \big]\\
=&\textstyle F(\boldsymbol{w}^{(t)}) - \eta (1 - 2e^{(t)}) \mathbb{E}_z\left[\frac{\boldsymbol{z}^\top \nabla F(\boldsymbol{w}^{(t)}) \boldsymbol{z}^\top \nabla F(\boldsymbol{w}^{(t)})}{|\boldsymbol{z}^\top \nabla F(\boldsymbol{w}^{(t)})|}\right]\\
&\textstyle + \mathbb{E}\big[\frac{\eta^2\boldsymbol{z}^\top \boldsymbol{H}_{\boldsymbol{w}^{(t)}}\boldsymbol{z}^\top}{2}\big],
\end{align*}
where $e^{(t)}$ is the probability that the aggregated sign differs from the true sign, the equality comes from the fact that $\boldsymbol{z}^\top \nabla F(\boldsymbol{w}^{(t)}) \boldsymbol{z}^\top \nabla F(\boldsymbol{w}^{(t)}; \mathcal{B}) = -\boldsymbol{z}^\top \nabla F(\boldsymbol{w}^{(t)}) \boldsymbol{z}^\top \nabla F(\boldsymbol{w}^{(t)})$ holds with probability $e^{(t)}$. Since $\boldsymbol{z}$ is a standard i.i.d. Gaussian random vector, $\frac{\boldsymbol{z}^\top \nabla F(\boldsymbol{w}^{(t)}) \boldsymbol{z}^\top \nabla F(\boldsymbol{w}^{(t)})}{|\boldsymbol{z}^\top \nabla F(\boldsymbol{w}^{(t)})|}$ will be half-Gaussian with a mean of $(\frac{2}{\pi})^{1/2}\|\nabla F(\boldsymbol{w}^{(t)})\|_2$.
For the last term in the right-hand side, based on \cref{ass:r_effective_rank} we have 
\begin{align}\label{eq:feedsign}
\mathbb{E}[F(\boldsymbol{w}^{(t+1)})]
\leq &\textstyle F(\boldsymbol{w}^{(t)}) - \eta(\frac{2}{\pi})^{1/2} \|\nabla F(\boldsymbol{w}^{(t)})\|_2 \nonumber \\
&\textstyle + 2 \eta e^{(t)}(\frac{2}{\pi})^{1/2}\|\nabla F(\boldsymbol{w}^{(t)})\|_2 + \frac{\eta^2Lr}{2}.
\end{align}
To obtain a convex objective function for the optimality gap optimization problem, we need to decouple $e^{(t)}$ from the multiplication. Using the arithmetic inequality, for $\forall \theta > 0$ we have
\begin{align}
    \textstyle 2 \eta e^{(t)} (\frac{2}{\pi})^{1/2} \|\nabla F(\boldsymbol{w}^{(t)})\|_2 \leq \theta (e^{(t)})^2 + \frac{2 \eta^2}{\theta\pi} \|\nabla F(\boldsymbol{w}^{(t)})\|_2^2. \label{eq:52}
\end{align}
Substituting \cref{eq:52} into \cref{eq:feedsign} yields \begin{align}\label{eq:53}
\mathbb{E}[F(\boldsymbol{w}^{(t+1)})]
\leq &\textstyle F(\boldsymbol{w}^{(t)}) - \eta(\frac{2}{\pi})^{1/2} \|\nabla F(\boldsymbol{w}^{(t)})\|_2 \\
& + \textstyle \theta (e^{(t)})^2 + \frac{2 \eta^2}{\theta\pi} \|\nabla F(\boldsymbol{w}^{(t)})\|_2^2 + \frac{\eta^2Lr}{2}. \nonumber
\end{align}
For $\|\nabla F(\boldsymbol{w}^{(t)})\|_2$ in \cref{eq:53}, take $0 < \eta \leq \min\{\frac{2^{3/2} \pi^{1/2} \theta }{5 \|\nabla F(\boldsymbol{w}^{(t)})\|_2}\}$ and apply \cref{ass:pl} yields 
\begin{align}
    &\textstyle - \eta(\frac{2}{\pi})^{1/2} \|\nabla F(\boldsymbol{w}^{(t)})\|_2 + \frac{2 \eta^2}{\theta\pi} \|\nabla F(\boldsymbol{w}^{(t)})\|_2^2 \nonumber \\
    \leq& \textstyle -\frac{\eta^2}{2 \theta \pi} \|\nabla F(\boldsymbol{w}^{(t)})\|_2^2 \nonumber \\
    \leq& \textstyle - \frac{M \eta^2}{\theta \pi}(F(\boldsymbol{w}^{(t)}) - F^*) \label{eq:54}
\end{align}
Substituting \cref{eq:54} into \cref{eq:feedsign}, and subtracting $F^*$ in both sides of the resulting inequality, we have \begin{align}\label{eq:55}
    &\textstyle \mathbb{E}[F(\boldsymbol{w}^{(t+1)})] - F^*
    \nonumber\\
    \leq &\textstyle (1 - \frac{M \eta^2}{\theta \pi})(F(\boldsymbol{w}^{(t)}) - F^*) + \theta (e^{(t)})^2 + \frac{\eta^2 L r}{2}
    \\
    \leq & \tilde{A} G(\boldsymbol{w}^{t})+\theta(e^{(t)})^2+r 
    \\
    \leq & \textstyle \tilde{A}^T G^{(0)} + \sum_{t=1}^T \tilde{A}^{-t}\big(\theta \cdot (e^{(t)})^2 + r\big)
\end{align}
where the second inequality is obtained by setting $0 < \eta \leq (\frac{2}{L})^{1/2}$ as well as the definition of $\tilde{A}$, and the third inequality is obtained by telescoping the third line.

\subsection{Proof of \cref{upper-bound-et}}
\label{pf:upper-bound-et}
Denote $s_{\text{true}} \triangleq \text{sign}\{\nabla F(\boldsymbol{w}^{(t)})^\top \boldsymbol{z}^{(t)}\}$. For simplicity, we assume that local datasets are independent and identically distributed, which means that $e_k^{(t)}$s are identical across all clients. Nevertheless, the subscript is preserved to distinguish them from the total sign reversing probability $e^{(t)}$.

Suppose $n_k^{(t)} = z_k^{(t)} = 0$, then from (\cref{eq:rec_digital}) we have 
\begin{align}
    y^{(t)} = c^{(t)} (1 - 2w) s_{\text{true}},
\end{align}
where $w \sim \text{Binomial}(K, e_k^{(t)})$. For this case, it holds 
\begin{align}
    \textstyle \mathbb{E}[y^{(t)}] &= c^{(t)}K(1 - 2e_k^{(t)})s_{\text{true}}, \\
    \textstyle \mathbb{V}[y^{(t)}] &= 4(c^{(t)})^2 K e_k^{(t)}(1 - e_k^{(t)}).
\end{align}
Now suppose the noise in $y_k^{(t)}$ is nonzero. Recall that the effective noise is independent of the transmitted signal, therefore, we have 
\begin{align}
    \textstyle \mathbb{V}[y^{(t)}] &= 4(c^{(t)})^2 K e_k^{(t)}(1 - e_k^{(t)}) + (m_k^{(t)})^2.
\end{align}
Further invoking \cref{eq:estp}, we have 
\begin{align}
    \textstyle \mathbb{E}[\hat{p}^{(t)}] &= \textstyle (1 - 2e_k^{(t)}) \cdot s_{\text{true}}, \\
    \textstyle \mathbb{V}[\hat{p}^{(t)}] &= \textstyle \frac{4 e_k^{(t)}(1 - e_k^{(t)})}{K} + \frac{(m^{(t)})^2}{(K c^{(t)})^2}.
\end{align}
Combining \cref{eq:def_et} and the definition of $s_{\text{true}}$, we have 
\begin{align}
    \textstyle e^{(t)} = \mathbb{P}(s_{\text{true}} \cdot \hat{p}^{(t)} < 0),
\end{align}
Recall the Chebyshev-Cantelli inequality \cite{cantelli-chebyshev}, namely, 
\begin{align}
\mathbb{P}(X < 0) = \mathbb{P}(X - \mathbb{E}[X] \leq \lambda)& \leq \mathbb{V}[X] / (\mathbb{V}[X] + \lambda^2), \nonumber\\
& \text{with } \lambda = -\mathbb{E}[X]
\end{align}
Let $X$ be set as $s_{\text{true}} \cdot \hat{p}^{(t)}$, then the Chebyshev-Cantelli inequality leads to 
 \begin{align}\label{eq:upper_bound_et}
    \textstyle (e^{(t)})^2 \leq e^{(t)} \leq & \textstyle \frac{4Ke_k^{(t)}(1-e_k^{(t)}) + \frac{(m^{(t)})^2}{(c^{(t)})^2}}{4Ke_k^{(t)}(1-e_k^{(t)}) + \frac{(m^{(t)})^2}{(c^{(t)})^2} + K^2(1-2e_k^{(t)})^2}
    \nonumber\\
    \leq & \textstyle 
    \frac{4Ke_0(1-e_0) + \frac{(m^{(t)})^2}{(c^{(t)})^2}}{4Ke_0(1-e_0) + \frac{(m^{(t)})^2}{(c^{(t)})^2} + K^2(1-2e_0)^2}
\end{align}
where the last inequality has invoked \cref{eq:e_0} (i.e., $e_k^{(t)}\leq e_0$) as well as the fact that this fractional function is monotonically increasing in the regime of $0 < e_k^{(t)} < 1/2$. 

{\subsection{Proof of \cref{thm:2}}\label{pf:P3}}

Recall that $e^{(t)}$ is upper bounded in \cref{upper-bound-et}, replacing $e^{(t)}$ by this upper bound yields
\begin{align}
\tilde{A}^{-t} &\textstyle \eta \theta (e^{(t)})^2\le 
\tilde{A}^{-t} \eta \theta \frac{B_1 + ( \sum_{k=1}^K (\sigma_k^{(t)})^2 + \frac{N_0}{(c^{(t)})^2} )}
{B_1 + B_2 + ( \sum_{k=1}^K (\sigma_k^{(t)})^2 + \frac{N_0}{(c^{(t)})^2} )}\\
&\textstyle =\tilde{A}^{-t} \eta \theta\Big(1 - 
    \frac{B_2}{B_1+B_2 +( \sum_{k=1}^K (\sigma_k^{(t)})^2 + \frac{N_0}{(c^{(t)})^2} )}\Big)
\end{align}
where
\begin{align}
    B_1 = K^2 \left( 1 - 2e_0 \right), B_2 = 4Ke_0\left(1-e_0\right).
\end{align}
Using fundamental algebra, we have 
\begin{align}
   \textstyle m_k^{(t)} = (\sigma_k^{(t)})^2+N_0/(c_k^{(t)})^2,n_k^{(t)}=(h_k^{(t)})^{2}(c_k^{(t)})^{-2}.
\end{align}
Thus, the original variables can be written as
\begin{align}
    \textstyle (\sigma_k^{(t)})^2= m_k^{(t)}-N_0 n_k^{(t)}/ (h_k^{(t)})^{2}, (c_k^{(t)})^{2}=(h_k^{(t)})^{2}(n_k^{(t)})^{-1}.
    \label{eq:old_variables}
\end{align}
Substituting this into problem (P2), we obtain the following equivalent problem:  
\begin{align}
    \textstyle \min\limits_{\{c_k^{(t)}, \sigma_k^{(t)}\}} & \quad\textstyle\sum_{t=1}^T \tilde{A}^{-t} (1 - 
    \frac{B_2}{B_1+B_2 +( m^{(t)})})\tag{P3}\\
        \text{s.t.} &\textstyle \quad \sum_{t=1}^T \frac{2}{ m^{(t)}} \leq R_{\mathsf{dp}}(\epsilon, \delta), \nonumber\\
    &\textstyle\quad (1/n_k^{(t)})(1 + d(m_k^{(t)}-N_0 n_k^{(t)}/(h_k^{(t)})^{2})^2) \leq P, \ \forall t,\nonumber\\
    &\textstyle\quad  m_k^{(t)}-N_0 n_k^{(t)}/(h_k^{(t)})^{2} \ge 0, \ \forall t,\nonumber\\
    &  \quad m_k^{(t)} \ge 0, n_k^{(t)} \ge 0, \ \forall t.
\end{align}
which is a convex problem. To solve this problem, define the Lagrange function as
\begin{align}
   &\textstyle  \mathcal{L}=\sum\limits_{t=1}^T \tilde{A}^{-t} (1 - 
    \frac{B_2}{B_1+B_2 + m^{(t)}}) + \zeta( \sum_{t=1}^{T} \frac{2}{m_k^{(t)}} - R_{\mathsf{dp}}(\epsilon, \delta)) \notag\\
    &\textstyle+ \sum_{t=1}^T (\xi^{(t)}( N_0 n_k^{(t)}/(h_k^{(t)})^{2}  - m_k^{(t)} )+\notag\\
    &\textstyle \sum\nolimits_{t=1}^T \beta^{(t)}( (1 + d(m_k^{(t)}) - ( d N_0 /(h_k^{(t)})^{2} +  P )n_k^{(t)} )
\end{align}
where \(\zeta\ge 0\), \(\beta^{(t)}\ge 0\) and \(\xi^{(t)}\ge 0\) are the Lagrange multipliers associated respectively with the DP constraint, transmit power constraints, and non-negative parameter constraints. The KKT condition can be given accordingly as 
\begin{align}
    &\textstyle \frac{\partial \mathcal{L}}{\partial m_k^{(t)}} 
    = \tilde{A}^{-t} (
    \frac{B_2}{B_1+B_2 + (m^{(t)})})^2 - 2\zeta/(m_k^{(t)})^{2}\notag \\
    &\qquad\qquad\qquad\qquad\qquad\qquad\qquad  +\beta^{(t)}d -\xi^{(t)} = 0. \label{eq:KKT1}\\
    &\textstyle \frac{\partial \mathcal{L}}{\partial n_k^{(t)}} 
    = -\beta^{(t)}(dN_0/(h_k^{(t)})^2 +P) \notag \\
   &\qquad\qquad\qquad\qquad\qquad\qquad +\xi^{(t)}N_0/(h_k^{(t)})^2 = 0, \label{eq:KKT2} \\
   &\textstyle\zeta(\sum_{t=1}^T \frac{2}{m_k^{(t)}} - R_{\rm dp}(\epsilon,\delta))=0, \label{eq:KKT3} \\
    &   \beta^{(t)}\Big(1 + d(m_k^{(t)}) -( dN_0 / (h_k^{(t)})^2 + P )n_k^{(t)}\Big)=0, \label{eq:KKT4}\\
    &\xi^{(t)}(m_k^{(t)}-N_0 n_k^{(t)}/(h_k^{(t)})^{2})=0, \label{eq:KKT5}\\
    &\textstyle \sum_{t=1}^T \frac{2}{m_k^{(t)}} - R_{\mathsf{dp}}(\epsilon,\delta) \le 0, \label{eq:KKT6}\\
    &1 + d(m_k^{(t)}) - ( dN_0/(h_k^{(t)})^2 + P )n_k^{(t)} \leq 0, \label{eq:KKT7}\\
    &N_0 n_k^{(t)}/(h_k^{(t)})^2 - m_k^{(t)}\le 0. \label{eq:KKT8}
\end{align}
\cref{eq:KKT2} immediately implies 
\begin{align}
    \textstyle \xi^{(t)}= \beta^{(t)}\frac{(dN_0/(h_k^{(t)})^2 + P)}
    {N_0/(h_k^{(t)})^2}\label{eq:KKT2-1}
\end{align}
Plugging this equality into \cref{eq:KKT1} and \cref{eq:KKT5} yields 
\begin{align}
 & \textstyle   \tilde{A}^{-t} \frac{B_2^2}{(B_1+B_2+ m_k^{(t)})^2} - \zeta\frac{2}{( m_k^{(t)})^2} - \beta^{(t)}\frac{P(h_k^{(t)})^2}{N_0} = 0 \label{eq:KKT1-1}\\
 & \textstyle \beta^{(t)}\frac{d ( \sqrt{N_0}/h_k^{(t)} )^2 + P}{( \sqrt{N_0}/h_k^{(t)} )^2}( \frac{N_0}{( h_k^{(t)})^2} n_k^{(t)} - m_k^{(t)} ) = 0 \label{eq:KKT5-1}
\end{align}
Combining \cref{eq:KKT5-1} and \cref{eq:KKT4}, we get the following equation
\begin{align}
  \textstyle \beta^{(t)}( 1 - \frac{P (h_k^{(t)})^2}{N_0} m_k^{(t)}) = 0 
    \label{eq:KKT_beta}
\end{align}
With a fixed $m_k^{(t)}$, we can the lower bound and the upper bound of $n_k^{(t)}$ from \cref{eq:KKT7} and \cref{eq:KKT8}, respectively. To ensure that the lower bound is smaller than the upper bound, $m_k^{(t)}$ should satisfy 
\begin{align}
   \textstyle m_k^{(t)} \ge \frac{N_0}{(h_k^{(t)})^2 P} \label{eq:m-bound-1}
\end{align}
In this case, the power is fully utilized for transmitting the local gradient. Furthermore, from \cref{eq:KKT_beta}, we have the equality \(\beta^{(t)}=0\) if follows \cref{eq:m-bound-1}. With $\beta^{(t)}=0$, \cref{eq:KKT1-1} becomes
\begin{align}
   \textstyle  (\tilde{A}^{-t} B_2^2 - 2\zeta ) ( m_k^{(t)} )^2 - 4\zeta( B_1 + B_2 ) m_k^{(t)} \notag \\- 2\zeta ( B_1 + B_2 )^2 = 0 \label{eq:m-quadratic}
\end{align}
Since $m_k^{(t)}>0$, the positive root of the quadratic equation \cref{eq:m-quadratic} should be the desired solution: 
\begin{align}
   \textstyle m_k^{(t)}= \frac{( B_1 + B_2) ( 4\zeta + \sqrt{8 \tilde{A}^{-t} B_2^2 \zeta} )}{2( \tilde{A}^{-t}B_2^2 - 2\zeta )} \label{eq:m-bound-2}
\end{align}
Since $m_k^{(t)}$ should also satisfy the constraints \cref{eq:KKT7} and \cref{eq:KKT8}, therefore we have 
\begin{align}
   \textstyle  m_k^{(t)} = \max \Big\{ \frac{ N_0}{(h_k^{(t)})^2 P}, 
    \frac{( B_1 + B_2 ) ( 4\zeta + \sqrt{8 \tilde{A}^{-t} B_2^2 \zeta} )}{2 ( \tilde{A}^{-t}B_2^2 - 2\zeta)} \Big\} \label{bound-mkt}
\end{align}
Since $m_k^{(t)}$ should also satisfy \cref{eq:KKT6}, we need to ensure the second term in the right-hand side of (\cref{bound-mkt}) satisfies \cref{eq:KKT6}. This can be achieved by searching for a proper $\zeta$ via bisection search. In particular, from \cref{eq:KKT3} we know that $\zeta$ should be $0$ if $\sum_{t=1}^T 2(h_k^{(t)})^2P/N_0 <R_{\mathrm{dp}}(\epsilon, \delta)$. After obtaining $m_k^{(t)}$, the value of $n_k^{(t)}$ can be obtained by using \cref{eq:KKT7} and \cref{eq:KKT8}, namely, 
\begin{align}
   \textstyle \frac{1 + d m_k^{(t)}}{d (\sqrt{N_0} / (h_k^{(t)}))^2 + P}\le 
   n_k^{(t)}\le \frac{(h_k^{(t)})^2 m_k^{(t)}}{N_0}
   \label{eq:n-bound}
\end{align}
At last, since the optimal value of $m_k^{(t)}$ is known, and the range of $n_k^{(t)}$ is determined by  \cref{eq:n-bound}. According to \cref{eq:old_variables}, the feasible range of $\sigma_k^{(t)}$ and $c_k^{(t)}$ can be located. Recall that we prefer a solution with the minimum transmit power; therefore, let $n_k^{(t)}$ be set to the right-hand side of \cref{eq:n-bound}, which yields the desired result. 





\bibliographystyle{IEEEtran}
\bibliography{references}

\vfill

\end{document}